\documentclass[journal]{IEEEtran}
\IEEEoverridecommandlockouts

\usepackage{grffile}
\usepackage{amsmath,amssymb,amsbsy}
\usepackage{hyperref}
\usepackage{amsthm} 
\usepackage{color}
\usepackage{comment}
\usepackage{lipsum}
\usepackage{graphicx}
\usepackage[normalem]{ulem}
\usepackage{enumerate}
\usepackage{epstopdf}
\usepackage{bbm}
\usepackage{subcaption}

\newtheorem{theorem}{Theorem}[section]
\newtheorem{lemma}[theorem]{Lemma}
\newtheorem{definition}[theorem]{Definition}	
\newtheorem{proposition}[theorem]{Proposition}
\newtheorem{corollary}[theorem]{Corollary}	
\newtheorem{remark}[theorem]{Remark}	
\newtheorem{problem}[theorem]{Problem}
\newtheorem{example}[theorem]{Example}
\newtheorem{assumption}[theorem]{Assumption}

\newcommand{\kar}[1]{{\color{black} #1}}
\newcommand{\essup}{${\rm ess sup}$}
\newcommand{\x}{\mathbf{x}}
\newcommand{\y}{\mathbf{y}}
\newcommand{\eL}{\mathbf{L}}
\newcommand{\ga}{\boldsymbol{\gamma}}

\begin{document}

\title{Density Stabilization Strategies for Nonholonomic Agents on Compact Manifolds}

\author{Karthik Elamvazhuthi and Spring Berman%
\thanks{
KE and SB were partially supported by ONR Young Investigator Award N000141612605. KE was also partially supported by AFOSR MURI FA9550-18-1-0502 and  ONR Grant No. N000142012093.}
\thanks{Karthik Elamvazhuthi is with the Department of Mechanical Engineering, University of California, Riverside, CA 92521, USA 
		{\tt\small karthike@ucr.edu}.}
\thanks{Spring Berman is with the School for Engineering of Matter, Transport and Energy, Arizona State University, Tempe, AZ, 85281 USA
		{\tt\small spring.berman@asu.edu}.}
}

\maketitle

\begin{abstract} 
 	In this article, we consider the problem of stabilizing a class of degenerate stochastic processes, which are constrained to a bounded Euclidean domain or a compact smooth manifold, to a given target probability density. This stabilization problem arises in the field of swarm robotics, for example in applications where a swarm of robots is required to cover an area according to a target probability density. Most existing works on modeling and control of robotic swarms that use partial differential equation (PDE) models assume that the robots' dynamics are holonomic, and hence, the associated stochastic processes have generators that are elliptic. We relax this assumption on the ellipticity of the generator of the stochastic processes, and consider the more practical case of the stabilization problem for a swarm of agents whose dynamics are given by a controllable driftless control-affine system. We construct state-feedback control laws that exponentially stabilize a swarm of nonholonomic agents to a target probability density that is sufficiently regular. State-feedback laws can stabilize a swarm only to  target probability densities that are positive everywhere. To stabilize the swarm to probability densities that possibly have disconnected supports, we introduce a semilinear PDE model of a collection of interacting agents governed by a hybrid switching diffusion process. The interaction between the agents is modeled using a (mean-field) feedback law that is a function of the local density of the swarm, with the switching parameters as the control inputs. We show that under the action of this feedback law,  the semilinear PDE system is globally asymptotically stable about the given target probability density. \kar{The stabilization strategies with and without agent interactions are verified numerically for agents that evolve according to the Brockett integrator;  the strategy with interactions is additionally verified for agents that evolve according to an underactuated system on the sphere $S^2$. }
\end{abstract}


	\section{Introduction}

In recent years, there has been much work on the construction of decentralized control laws for multi-robot systems using {\it mean-field models} \cite{elamvazhuthi2019mean}, in which a large collective of agents is treated as a continuum. Mean-field based approaches for modeling and control have been used in the multi-agent control and swarm robotics literature for problems such as consensus \cite{valentini2017achieving}, flocking \cite{han2015styled}, task allocation  \cite{berman2009optimized,mayya2019closed}, and cooperative transport \cite{wilson2014design}. Similar problems have also been considered over the last two decades in the control and mathematics literature in the context of mean-field games \cite{lasry2007mean,huang2006large}, mean-field control \cite{fornasier2014mean,piccoli2015control,bonnet2019pontryagin}, and optimal transport \cite{benamou2000computational}.

One advantage of mean-field based approaches for decentralized control is that the constructed control laws are {\it identity-free}, that is, the control laws do not depend on the agents' identities. The identity-free nature of the control laws simplifies some aspects of their  implementation  on large swarms of homogeneous agents when compared to control laws that are identity-dependent. For instance,
suppose that a central supervisor observes the states of the agents via an overhead camera and uses these measurements to update  
state-feedback control laws, which it periodically broadcasts to the agents. If the control laws are identity-free, then the supervisor does not need to expend computational power to distinguish between individual agents. 
Another advantage of mean-field control approaches, from a theoretical point of view, is that as the number of agents tends to infinity, the mean-field behavior of the swarm is governed by a deterministic differential equation or difference equation model, 
even though each agent might exhibit stochasticity in its dynamics. Such models are more analytically tractable than models describing the dynamics of a large number of individual agents. 

In this article, we consider a mean-field stabilization problem motivated by coverage problems in multi-agent control. A classical approach to multi-agent coverage is described in \cite{cortes2004coverage}, which presents a distributed method for implementing Lloyd's algorithm for positioning  multiple agents in a domain according to a given probability density function. The stochastic task allocation problem considered in \cite{berman2009optimized}, where the goal is to stabilize a swarm of agents evolving according to a continuous-time Markov chain to a target distribution among a set of states (e.g., physical locations), can be viewed as a mean-field version of this coverage problem. Stochastic task allocation approaches have also been developed for swarms of agents that evolve according to discrete-time Markov chains \cite{accikmecse2015markov} and that follow control laws which depend on the local density of agents \cite{mather2011distributed,elamvazhuthi2018mean}. 
A drawback of Markov chain-based approaches to this problem 
is that the state space of the agents needs to be discretized beforehand.

In \cite{mesquita2008optimotaxis,mesquita2012jump}, the authors consider the problem of stabilizing a swarm of agents with general controllable dynamics to a target distribution. This approach has been extended to the case of holonomic agents on bounded domains \cite{elamvazhuthi2016coverage,elamvazhuthi2018bilinear} and compact manifolds without boundary \cite{elamvazhuthi2018nonlinear}. In these extensions, the control is either a diffusion coefficient or a velocity  field in the Fokker-Planck partial differential equation (PDE) \cite{risken1996fokker} that determines the spatio-temporal evolution of the probability density of the agents, each of which is governed by a reflected stochastic differential equation \cite{tanaka1979stochastic}. In \cite{biswal2019stabilization}, similar stochastic control laws are constructed for agents evolving according to a discrete-time deterministic nonlinear control system on a bounded Euclidean domain. 

An advantage of the stochastic coverage approaches in \cite{berman2009optimized,accikmecse2015markov,mesquita2008optimotaxis,mesquita2012jump,elamvazhuthi2016coverage,elamvazhuthi2018bilinear,elamvazhuthi2018nonlinear,biswal2019stabilization} over the classical coverage strategy in \cite{cortes2004coverage} is that they only require each agent's control action to depend on its own state. The stochasticity designed into the system ensures that the swarm reaches a target probability density. 
A disadvantage of these approaches is that agents do not stop switching between states even after the swarm has reached the target probability density, resulting in an unnecessary expenditure of energy. Moreover, a large number of agents is required for the swarm density to stabilize close to the target density. One way to resolve these issues is to design control laws that are functions of the local swarm density. Such density-dependent or {\it mean-field} feedback laws have been proposed for agents that evolve according to Markov chains on discrete state spaces  \cite{mather2011distributed,elamvazhuthi2018mean}, for agents that evolve according to ordinary differential equations on Euclidean domains \cite{eren2017velocity,krishnan2018distributed}, and for agents that evolve according to stochastic or ordinary differential equations on compact manifolds without boundary \cite{elamvazhuthi2018nonlinear}. The works \cite{eren2017velocity,krishnan2018distributed,elamvazhuthi2018nonlinear} assume that the agents are holonomic, and \kar{in order to achieve global asymptotic stability, the target distribution is required to be strictly positive everywhere on the domain.} 

In the context of these previous works, the \textbf{main contributions} of this article are the following:
	\begin{enumerate}
		\item  Extension of the stochastic multi-agent coverage approach developed by the authors in \cite{elamvazhuthi2016coverage,elamvazhuthi2018bilinear} to non-interacting agents with nonholonomic dynamics that evolve on a bounded subset of $\mathbb{R}^d$ or a compact manifold without boundary, given a target  probability density for the swarm that is bounded from below by a positive constant.
			\label{ct2}
		\item Development of a stochastic multi-agent coverage approach for interacting agents governed by the
		hybrid switching diffusion model introduced by the authors in \cite{elamvazhuthi2018optaut}. In this approach, a mean-field feedback law (i.e., a control law that depends on the local swarm density) is constructed to \kar{globally} asymptotically stabilize the swarm to any given target probability density, which  is not necessarily positive everywhere. 
			\label{ct3}
	\end{enumerate}
Contribution \ref{ct2} is partially motivated by the fact that extension of multi-agent control strategies designed for Euclidean state spaces to general manifolds \cite{sarlette2009consensus,bhattacharya2014multi} is important, given that many mechanical systems are naturally modeled on manifolds \cite{bullo2019geometric}. While the work \cite{mesquita2012jump}, as in Contribution \ref{ct2}, does consider agents with general nonholonomic dynamics, due to the assumption that the domain is unbounded, this work requires assumptions on the behavior of the target probability density at infinity. The extension of the coverage strategy presented in \cite{elamvazhuthi2016coverage,elamvazhuthi2018bilinear}  to the case of nonholonomic agents evolving on manifolds is complicated by the fact that the associated PDEs are not elliptic. Contribution \ref{ct3} improves over existing work on mean-field feedback laws \cite{eren2017velocity,krishnan2018distributed}, since our approach does not make strong assumptions on the regularity of solutions of the associated PDEs. Instead, we prove all the regularity required to enable the stability analysis, which makes the analysis much more technically involved. Moreover, we are able to stabilize a larger class of probability densities than those considered in \cite{elamvazhuthi2016coverage,elamvazhuthi2018bilinear,elamvazhuthi2018nonlinear,eren2017velocity,krishnan2018distributed}, which require that the target probability density is strictly bounded from below by a positive constant everywhere on the domain. A control law similar to the one described in Contribution \ref{ct3} was constructed by the authors in \cite{biswal2020stabilization} for a discrete-time control system. However, it is assumed in \cite{biswal2020stabilization} that the system is locally controllable within one time-step, and hence is fully actuated. Moreover, in contrast to the mean-field feedback laws constructed in \cite{eren2017velocity,krishnan2018distributed,elamvazhuthi2018nonlinear} in which the control input is the agents' velocity field, the control inputs that we design in Contribution \ref{ct3} are the transition rates of the hybrid switching diffusion process that describes the agents' dynamics. 

This article is organized as follows. In Section \ref{sec:not}, we establish
notation and provide some definitions that are used throughout
the article. In Section \ref{sec:subell}, we present and analyze the properties of the degenerate PDEs that describe the mean-field model in the case where the agents do not interact with one another. In Section \ref{sec:Disco}, we present a semilinear PDE mean-field model for stabilizing the density of a swarm of interacting agents and establish global asymptotic stability properties of the model. In Section \ref{sec:numsim}, we validate the control strategies presented in Sections \ref{sec:subell} and \ref{sec:Disco} with numerical simulations.

	\section{NOTATION}
	\label{sec:not}
	We denote the $n$-dimensional Euclidean space by $\mathbb{R}^n$. $\mathbb{R}^{n \times m}$ refers to the space of $n \times m$ matrices, and $\mathbb{R}_+$ refers to the set of non-negative real numbers. Given a vector $\mathbf{x}  \in \mathbb{R}^n$, $x_i$ denotes the $i^{th}$ coordinate value of $\mathbf{x}$. For a matrix $\mathbf{A} \in \mathbb{R}^{n \times m}$, $A_{ij} $ refers to the element in the $i^{th}$ row and $j^{th}$ column of $\mathbf{A}$. For a subset $B \subset \mathbb{R}^M$, ${\rm int}(B)$ refers to the interior of the set $B$. $\mathbb{C}$, $\mathbb{C}_-$, and $\bar{\mathbb{C}}_-$ denote the set of complex numbers, the set of complex numbers with negative real parts, and the set of complex numbers with non-positive real parts, respectively. $\mathbb{Z}_+$ refers to the set of positive integers.

	We denote by $\Omega$ an open, bounded, and connected subset of an $N$-dimensional smooth Riemannian manifold $M$ \cite{lee2001introduction,lee2018introduction} with a Riemannian volume measure $d\mathbf{x}$. The boundary of $\Omega$ is denoted by $\partial \Omega$. We denote by $\int_{\Omega}f(\mathbf{x})d \x$ the integral of a function $f:\Omega \rightarrow \mathbb{R}$ with respect to the Riemannian volume.
	
	For each $1\leq p< \infty$, we define $L^p(\Omega)$ as the Banach space of complex-valued measurable functions over the  set $\Omega$ whose absolute value raised to $p^{{\rm th}}$ power has finite integral. \kar{The space $L^1$ is equipped with the norm $\|f\|_1:= \int_{\Omega } |f(x)|dx$.} We define $L^{\infty}(\Omega)$ as the space of essentially bounded measurable functions on $\Omega$. The space $L^{\infty}(\Omega)$ is equipped with the norm
	$
	\|z\|_{\infty} =~ \essup_{\mathbf{x} \in \Omega} |z(\mathbf{x})|, 
	$
	where $ \essup_{\mathbf{x} \in \Omega}(\cdot)$ denotes the {\it essential supremum} attained by its argument over the domain $\Omega$.
	
	 For a given real-valued function $a \in L^{\infty}(\Omega)$, $L^2_a(\Omega)$ refers to the set of all functions $f$ such that the norm of $f$ is defined as
$
\|f\|_a: = (\int_\Omega|f(\mathbf{x})|^2a(\mathbf{x})d\mathbf{x} )^{{1/2}}< \infty.
$
We will always assume that the associated function $a$ is uniformly bounded from below by a positive constant, in which case the space $L^2_a(\Omega)$ is a Hilbert space with respect to the weighted inner product $\langle \cdot , \cdot \rangle_{a} :  L^2_a(\Omega) \times L^2_a(\Omega) \rightarrow \mathbb{R}$, given by
$
\langle f, g\rangle_{a} = \int_{\Omega} f(\mathbf{x})\bar{g}(\mathbf{x})a(\mathbf{x})d\mathbf{x}
$
for each $f,g \in L^2_a(\Omega)$, where $\bar{g}$ is the complex conjugate of the function $g.$ When $a =\mathbf{1}$, where $\mathbf{1}$ is the function that takes the value $1$ almost everywhere on $\Omega$, the space $L^2_a(\Omega)$ coincides with the space $L^2(\Omega)$. For a function $f \in L^2(\Omega)$ and a given constant $c$, we write $f \geq c$ to imply that $f$ is real-valued and $f(\mathbf{x})\geq c$ for almost every (a.e.) $\mathbf{x} \in \Omega$.

	Suppose $e^{Xt}$ is the flow generated by a vector field $X$. Then $X$ defines a differential operator on the set of smooth functions $C^{\infty}(M)$ through the following action,
	\begin{equation}
	\label{eq:diffop}
	(Xf)(\mathbf{x}) = \lim_{t \rightarrow 0}\frac{f(e^{tX}(\x)) - f(\mathbf{x})}{t}
	\end{equation}
	for all $\x \in \Omega$ and all $f \in C^{\infty}(\Omega)$.
	\kar{This is the differential geometric definition from \cite{lee2001introduction} of a vector field $X$ as an associated differential operator acting on the space of smooth functions. 
	
	Let $\mathcal{V} = \lbrace X_1,...,X_m \rbrace$, $m \leq N$, be a collection of smooth vector fields $X_i$, each defined as in \eqref{eq:diffop}.} Let $[X,Y]$ denote the Lie bracket of the vector fields $X$ and $Y$. \kar{We define $\mathcal{V}^0 =\mathcal{V}$. For each $i \in \mathbb{Z}_+$, we define in an iterative manner the set of vector fields $\mathcal{V}^i = \lbrace [X, Y]; ~X \in \mathcal{V}, ~Y \in \mathcal{V}^{j-1}, ~j=1,...,i \rbrace$.}   We will assume that the collection of vector fields $\mathcal{V}$ satisfies the {\it Chow-Rashevsky} condition \cite{agrachev2013control} (also known as {\it H\"{o}rmander's condition} \cite{bramanti2014invitation}), i.e., the Lie algebra generated by the vector fields $\mathcal{V}$, given by $\cup_{i =0}^r \mathcal{V}^i$, has rank $N$, for sufficiently large $r$.
	
A {\it horizontal curve} $\ga:[0,1] \rightarrow \Omega $ connecting two points $\x ,\y \in M$ is a Lipschitz curve in $\Omega$ for which there exist essentially bounded functions $a_i(t)$ such that 
	\begin{equation}
	\label{eq:ctraffhyp}
	\dot{\ga}(t) = \sum_{i = 1}^ma_i(t)X_i(\ga(t))
	\end{equation}
	for almost every $t \in [0,1]$, \kar{where $X_i \in \mathcal{V}$,} $\ga(0)=\x$, and $\ga(1)=\y$. Then $\mathcal{V}$  defines a distance $d :\Omega \rightarrow \mathbb{R}_{\geq 0}$ on $M$ as  
	\begin{eqnarray}
	d(\x,\y) = {\rm inf} ~ \lbrace \int_0^1 |\dot{\ga}(t)|dt;& \text{$\ga$ is a horizontal curve} \nonumber 
	\\ & \text{ connecting $\x$ and $\y$} \rbrace   \nonumber
	\end{eqnarray}
	
	\begin{definition}
	\label{def:epsdel}
	 The domain $\Omega \subset M$ is said to be $\epsilon-\delta$ if there exist $\delta >0$, $0<\epsilon \leq 1$ such that for any pair of points $\mathbf{p},\mathbf{q} \in \Omega$, if $d(\mathbf{p},\mathbf{q})  \geq \delta $, then there exists a continuous curve  $\ga:[0,T] \rightarrow \Omega$ such that $\ga(0) = \kar{\mathbf{p}}$, $\ga(T) = \kar{\mathbf{q}}$, and 
	 \begin{eqnarray}
	 & \int_0^1 |\dot{\ga}(t)|dt ~\geq~ \frac{1}{\epsilon} d(\kar{\mathbf{p},\mathbf{q}}) \nonumber \\
	& \hspace{-3mm} d(\mathbf{\kar{z}},\partial \Omega) ~\geq~ \epsilon ~{\rm min}(d(\kar{\mathbf{p},\mathbf{z}}),d(\kar{\mathbf{z},\mathbf{q}}))~~ \forall \kar{\mathbf{z}} \in \lbrace \ga(t):t\in [0,T] \rbrace  \nonumber 
	 \end{eqnarray}
    \end{definition}
	
	The metric $d$ on $\Omega$ is known as the {\it sub-Riemannian} or {\it Carnot-Caratheodory} metric 
\cite{agrachev2019comprehensive,bramanti2014invitation}. The topology induced by this metric on $d$ coincides with the usual Riemannian metric. We will assume that the radius $r (\Omega)$ of $\Omega$, given by $r (\Omega)= {\rm sup} \lbrace d(\x,\y); \x,\y \in M  \rbrace$, is finite. 
	
	 Given $a \in L^{\infty}(\Omega)$, with $a \geq c$ for a positive parameter $c>0$, we define the {\it weighted horizontal Sobolev space} $WH^1_a(\Omega) = \big \lbrace f \in L^2(\Omega): X_i(af) \in L^2(\Omega) \text{ for } 1 \leq i \leq m \big \rbrace$. We equip this space with the weighted horizontal Sobolev norm $\|\cdot\|_{WH^1_a}$, given by $\|f\|_{WH^1_a} = \Big( \|f\|^2_{2} + \sum_{i=1}^n\|  X_i(af)\|^2_{2}\Big)^{1/2} \nonumber$ for each $f \in WH^1_a(\Omega)$. Here, the derivative action of $X_i$ on a function $f$ is to be understood in the distributional sense. When $a=\mathbf{1}$, where $\mathbf{1}$ is the constant function that is equal to $1$ everywhere, we will denote $WH_a^1$ by $WH^1$.  
	
	Let $X$ be a Hilbert space with the norm $\|\cdot\|_X$. The space $C([0,T];X)$ consists of all continuous functions $u:[0,T] \rightarrow X$ for which  $ \|u\|_{C([0,T];X)} := 
 \max_{0 \leq t \leq T} \|u(t)\|_X ~<~ \infty$. If $Y$ is a Hilbert space, then $\mathcal{L}(X,Y)$ will denote the space of linear bounded operators from $X$ to $Y$. We will also use the multiplication operator $\mathcal{M}_a : L^2(\Omega) \rightarrow L^2(\Omega)$, defined as $(\mathcal{M}_a u)(\mathbf{x}) = a(\mathbf{x}) u(\mathbf{x})$ for a.e. $\mathbf{x} \in \Omega$ and each $u \in L^2(\Omega)$.
	
	We will need an appropriate notion of a solution of the PDEs considered in this paper. Toward this end, let $A$ be a closed linear operator that is densely defined on a subset $\mathcal{D}(A)$, the domain of the operator, of a Hilbert space $H$. We will define ${\rm spec}(A)$ as the set $\lbrace \lambda \in \mathbb{C}:\lambda\mathbb{I} -A $ is not invertible in $X \rbrace$, where $\mathbb{I}$ is the identity map on $X$. If $A$ is a bounded operator, then $\|A\|_{op}$ will denote the operator norm induced by the norm defined on $H$. From \cite{engel2000one}, we have the following definition.
	
	\begin{definition} 
		For a given time $T>0$, a \textbf{mild solution} of the ODE
		\begin{equation} 
		\label{eq:absode}
		\dot{u}(t) = Au(t); \hspace{2mm} u(0) = u_0 \in H
		\end{equation}
		is a function $u \in C([0,T];X)$ such that $u(t) = u_0 + A\int_0^tu(s)ds$ for each $t \in [0, T]$. 
	\end{definition}
	Under appropriate conditions satisfied by $A$, the mild solution of a PDE is given by a {\it strongly continuous semigroup} of linear operators, $(\mathcal{T}(t))_{t\geq 0}$, that are {\it generated} by the operator $A$ \cite{engel2000one}. 
	

	\begin{definition}
		A strongly continuous semigroup of linear operators $(\mathcal{T}(t))_{t \geq 0}$ on a Hilbert space $X$ is called \textbf{positive} if $u \in X$ such that $u \geq 0$ implies that $\mathcal{T}(t)u \geq 0$ for all $t \geq 0$.
	\end{definition}

\section{Stabilization without Agent Interactions}
\label{sec:subell}
Given the definitions in Section \ref{sec:not}, consider the following {\it reflected} stochastic differential equation (SDE) \cite{pilipenko2014introduction} constrained to a domain $\Omega \subseteq M$: 
\begin{eqnarray}
	\label{eq:SDEhol}
	d\mathbf{Z}(t) &=& \sum_{i=1}^m u_i(\mathbf{Z}(t))X_idt +\sqrt{2}\sum_{i=1}^m v_i(\mathbf{Z}(t))X_i\circ dW_i  \nonumber \\ &~& +	~\mathbf{n}(\mathbf{Z}(t))d\psi(t), \nonumber \\ 
	\mathbf{Z}(0) &=& ~ \mathbf{Z}_0,
\end{eqnarray}
where $\psi(t) \in \mathbb{R}$ is called the {\it reflecting function} or {\it local time} \cite{pilipenko2014introduction}, a stochastic process that constrains $\mathbf{Z}(t)$ to the domain $\Omega$; $\mathbf{n}(\mathbf{x})$ is the normal to the boundary at $\mathbf{x} \in \partial \Omega$; $W_i$ are  $m$ copies of the one-dimensional Wiener process; and $u_i$ and $v_i$ are $m$ feedback laws. 
In the above SDE \eqref{eq:SDEhol}, the notation $\circ$ is used to mean that the SDE should be interpreted in the {\it sense of Stratonovich} \cite{karatzas1998brownian}.
Let $y(\mathbf{x},t)$ denote the  probability density of the random variable $\mathbf{Z}(t)$, defined as $\mathbb{P}(\mathbf{Z}(t) \in A) = \int_{A} y(\mathbf{x},t)d\mathbf{x}$. In this section, we consider the following control problem:
\begin{problem}
	\label{prb:prb1}
Given a target probability density $f$ on $\Omega$, design control laws $u_i(\mathbf{x})$ and $v_i(\mathbf{x})$ in Eq. \eqref{eq:SDEhol}
such that the probability density $y(\mathbf{x},t)$of the 
stochastic process $\mathbf{Z}(t)$, which evolves according to Eq. \eqref{eq:SDEhol}, 
converges asymptotically to $f$. 
\end{problem}

The motivation for this problem comes from stochastic coverage applications in swarm robotics that are framed as follows. Let the random variable $\mathbf{Z}_j(t)$, $j \in \lbrace 1,...,N_p\rbrace$, denote the position of the $j^{th}$ robot in a swarm of $N_p$ robots at time $t$. This position evolves according to Eq. \eqref{eq:SDEhol}, in which $u_i$ and $v_i$ are control laws that govern each robot's motion. 
 \kar{Since each} robot follows the same control laws $u_i$ and $v_i$, 
 the random variables $\mathbf{Z}_i(t)$ are independent and identically distributed. Then, denoting by $\delta_\mathbf{x}$ the delta distribution at $\mathbf{x} \in M$, the {\it empirical distribution} $\frac{1}{N_p}\sum_{j=1}^{N_p}\delta_{\mathbf{Z}_{j}(t)}$, which represents the distribution of the robots in space, converges to the density $y(\mathbf{x},t)$ as $N_p \rightarrow \infty$ due to the {\it law of large numbers}. 

The stabilization problem \ref{prb:prb1} has been considered by the authors in \cite{elamvazhuthi2016coverage} for the case where the system  is {\it holonomic} and  the vector fields $X_i = \frac{\partial}{\partial x_i}$ are the standard coordinate vector fields. The goal in this section is to extend the results in \cite{elamvazhuthi2016coverage} to the general case where the number of vector fields $X_i$ is possibly less than the dimension $N$ of the state space $M$. Such density stabilization problems were first considered for the case where the domain $\Omega$ is the whole of the Euclidean space $\mathbb{R}^n$ in \cite{mesquita2008optimotaxis,mesquita2012jump}. When time is discrete, and the system is controllable in one time step, this problem has been considered in \cite{biswal2019stabilization}.

 The main difficulty in extending the results from \cite{elamvazhuthi2016coverage} is that when the number of control vector fields, $m$, is less than the dimension of the state space $M$, the {\it generator}
 $ \sum_{i = 1}^m (v_iX_i)^2 + u_iX_i$ of the stochastic process $\mathbf{Z}(t)$ is not {\it elliptic}, which makes standard results in the literature on parabolic PDEs 
 inapplicable. In particular, let $A =  \sum_{i = 1}^m (v_iX_i)^2$. The associated probability density $y(\mathbf{x},t)$ of the process $\mathbf{Z}(t)$ evolves according to the PDE 
\begin{eqnarray}
\label{eq:Mainsys1hol} 
&y_t = A^* y-  \nabla \cdot (\sum_{i=1}^m{u}_i(\mathbf{x})X_i y) & ~~in  ~~ \Omega \times [0,T] \nonumber \\ 
&y(\cdot,0) = y^0 & ~~in ~~ \Omega 
\end{eqnarray}
with {\it zero flux} boundary conditions, where $\nabla \cdot$ denotes the divergence operation with respect to the measure $d\mathbf{x}$, and $A^*$ is the adjoint of the operator $A$. The stabilization problem \ref{prb:prb1} is an open-loop control problem for the PDE \eqref{eq:Mainsys1hol}, in which the goal is to stabilize the solution $y(\mathbf{x},t)$ of \eqref{eq:Mainsys1hol} to a target function $f$. 

The operator $A$ is not elliptic in general, but only {\it hypoelliptic}. Particularly, if $f \in C^{\infty}_0(\Omega)$ has compact support $K$, then, due to the Chow-Rashevsky's Lie rank condition, if $u$ is a function on $\Omega$ such that $A u =f$, then $u$ is smooth on $K$  \cite{bramanti2014invitation}. Using this property of $A$, we will extend the stabilization results of \cite{elamvazhuthi2016coverage,elamvazhuthi2018bilinear} to the case where the agents have nonholonomic dynamics.

First, formally, we will provide a number of candidate control laws that are solutions to Problem \ref{prb:prb1}. Given these control laws, in Section \ref{sec:StAnalysis} we will present a stability analysis of a class of PDEs that coincide with \eqref{eq:Mainsys1hol} when the operators $X_i$ are formally skew-adjoint with respect to the volume form $d\mathbf{x}$, that is, $X_i^* = -X_i$. 

Suppose that the operators $X_i$ are formally skew-adjoint.  Let $f \in W^{1,\infty}(\Omega)$ be a positive function that is bounded from below by a positive number and for which $\int_{\Omega}f(\mathbf{x})d\mathbf{x}=1$. If we set $u_i(\cdot) = X_ig/g$ and $v_i(\cdot) = 1$ for each $i \in \lbrace 1,...,m \rbrace$ and all $t \geq 0$, then the PDE \eqref{eq:Mainsys1hol} becomes  
\begin{eqnarray}
	\label{eq:Mainsys1hol2} 
	&y_t = \sum_{i=1}^m X_i^2y -  \nabla \cdot (\sum_{i=1}^m \frac{X_if}{f}X_i y)~ & ~in  ~~ \Omega \times [0,T] \nonumber \\ 
	&y(\cdot,0) = y^0 & ~in ~~ \Omega 
\end{eqnarray}
Let $\nabla_H$ be the {\it horizontal gradient} operator, which maps functions to vector fields and is defined as $\nabla_H g = \sum_{i=1}^m (X_i g) X_i$
When the manifold $M$ is a {\it Lie group} $G$ and it is {\it unimodular}, i.e., the left- and right-Haar measures \cite{lee2001introduction} coincide, then we have that $\nabla \cdot \nabla_H (\cdot)= \sum_{i=1}^mX_i^2$ \cite{agrachev2009intrinsic}. Hence, if we set $y = f$, then 
\begin{equation}
\sum_{i=1}^mX_i^2  y -  \nabla \cdot \left(\sum_{i=1}^m \frac{X_if}{f}X_i y \right)=  \sum_{i=1}^mX_i^2 g - \nabla \cdot (\nabla_Hg)=0
\end{equation}

Thus, $f$ is an equilibrium solution of the PDE \eqref{eq:Mainsys1hol2}. Our goal in Section \ref{sec:StAnalysis} will be to show that $f$ is the globally exponentially stable equilibrium solution of PDE \eqref{eq:Mainsys1hol2} on the the set of square-integrable probability densities.  
 
Let $a = \frac{1}{f}$. Then the operator $\sum_{i=1}^m X_i^2 -  \nabla \cdot (\sum_{i=1}^m \frac{X_if}{f}X_i )$ can be alternatively expressed as  $ \nabla \cdot (\frac{1}{a(\mathbf{x})}\nabla_H(a(\mathbf{x}) \cdot))$.  Similarly, one can also consider the feedback laws $u_i = 0$ and $v_i = \frac{1}{f}$. In this case, the corresponding operator of interest is given by
\begin{equation}
\label{eq:diffgena}
A^* =  \nabla \cdot ({a(\mathbf{x})}\nabla_H(a(\mathbf{x}) \cdot))
\end{equation}  
This control law is similar to the one presented in \cite{elamvazhuthi2016coverage} for the case of holonomic agents, where instead of the Stratonovich integral, we considered the Ito integral, and the resulting generator was of the form  $\nabla \cdot (\nabla_H(a(\mathbf{x}) \cdot))$.

\subsection{Stability Analysis} \label{sec:StAnalysis}

The preceding discussion motivates us to study stability properties of PDEs associated with a class of hypoelliptic operators that have a given probability density as their equilibrium solution. In this section, we will provide a semigroup theoretic analysis of a class of such PDEs.  
There have been a number of works on semigroups generated by hypoelliptic operators on manifolds without boundary \cite{jerison1986estimates}, or manifolds with boundary under the Dirichlet boundary \cite{varopoulos2008analysis}. Due to the term $a(\mathbf{x})$, the operators that we consider are more general than those in   \cite{varopoulos2008analysis}. There has also been work on long-term behavior of hypoelliptic diffusions to uniform distributions for the special case of Carnot groups \cite{rossi2013large} and to more general equilibrium distributions, as well as on more general state spaces, for certain examples of hypoelliptic diffusions using log-Sobolev inequalities \cite{baudoin2012log}. In comparison, our results hold for a more general class of degenerate diffusions by establishing a spectral gap for the generator, which can be guaranteed to exist for equilibrium distributions that are bounded uniformly from above and below by positive numbers.

Before we present our stability analysis, we give some more preliminary definitions.
Given $a,b \in L^{\infty}(\Omega)$ such that $a \geq c$ and $b\geq c$ for some positive constant $c$, and $\mathcal{D}(\omega^b_a) = WH^1_a(\Omega)$, we define the sesquilinear form $\omega^b_a:\mathcal{D}(\omega^b_a) \times \mathcal{D}(\omega^b_a) \rightarrow \mathbb{C}$ as
\begin{equation}
	\label{eq:Clpgenform1hol}
	\omega^b_a(u,v) = \sum_{i= 1}^m\int_{\Omega}  b(\mathbf{x}) X_i ( a (\mathbf{x}) u(\mathbf{x})) \cdot X_i ( a (\mathbf{x})\bar{v}(\mathbf{x}))d\mathbf{x}
\end{equation} 
for each $u \in \mathcal{D}(\omega^b_a)$. 
We associate with the form $\omega^b_a$ an operator $A_a^b :\mathcal{D}(A_a^b)  \rightarrow  L^2_a(\Omega)$, defined as $ A_a^bu = v$, if $\omega_a^b(u,\phi) = \langle v , \phi \rangle_a $ for all $\phi \in \mathcal{D}(\omega^b_a)$ and for all $u \in \mathcal{D}(A_a^b) = \lbrace g \in \mathcal{D}(\omega_a^b): ~ \exists f \in L^2_a(\Omega) ~ \text{s.t.} ~  \omega_a^b(g,\phi)= \langle f, \phi \rangle_a ~ \forall \phi \in \mathcal{D}(\omega^b_a) \rbrace$. When the $X_i$ are formally skew-adjoint, the operator $A_a^b$ is a weak formulation of the the second-order partial differential operator $\sum_{i =1}^m X^*_i (b(\mathbf{x}) X_i (a(\mathbf{x} ) ~ \cdot ~))$.  An advantage of using this weak formulation of the operator, rather than  a strong formulation, is that one does not need to establish that the domain $\mathcal{D}(A_a^b)$ contains twice (weakly) differentiable functions, which might not be true in general, given the very weak regularity assumed on the boundary of the domain $\Omega$ \cite{jerison1989functional}. Formally, we will be studying properties of the following PDE, 
   \begin{eqnarray}
	\label{eq:Mainsysan} 
	&y_t = \sum_{i=1}^m X_i^*(b(\mathbf{x})X_i( a(\mathbf{x})y)) ~~ in  ~~ \Omega \times [0,T], 
\end{eqnarray}

Note that for $b=1/a$, $b=a$, and $b =\mathbf{1}$, we recover the operators introduced in the previous subsection.

The proofs of the results in this section follow closely, almost verbatim, to those for the elliptic case considered by the authors in \cite{elamvazhuthi2018bilinear}. Therefore, due to space limitations, we only include the proof of the stability result, which 
is stated in Theorem \ref{thm:stablin}. These results will be used extensively in Section \ref{sec:Disco}, where we consider the case in which the agents have local interactions with one another. The main technical difference between the proofs in \cite{elamvazhuthi2018bilinear} and the proofs of the results presented in this section is that here, we use the horizontal Sobolev spaces $WH^1({\Omega})$ to establish semigroup generation properties of the generator $A$, instead of the classical Sobolev space $H^1(\Omega)$. Due to the bracket generating property of the vector fields $\mathcal{V}$, it is known that the space $WH^1({\Omega})$ has many properties similar to the classical Sobolev space $H^1(\Omega)$ \cite{garofalo1996isoperimetric,nhieu2001neumann}.

We will need the following assumption for the results in  this section to hold true.
\begin{assumption}
\label{asmp1}
Only one of the following conditions holds:
	\begin{enumerate}
		\item The domain $\Omega$ is $\epsilon-\delta$  and  $M=\mathbb{R}^n$. \label{cond1}
		\item The manifold $M$ is compact and without a boundary. \label{cond2}
		\item The boundary $\partial \Omega $ of the domain $\Omega$ is $C^1$ and \kar{${\rm span} \{\mathcal{V} \} = T_\mathbf{x}M$} for all $\mathbf{x} \in \Omega$. \label{cond3}
		\end{enumerate}
\end{assumption}

\begin{lemma}
	\label{opprop1hol}
\begin{enumerate}
	\item 	The operator $A_a^b:\mathcal{D}(A_a^b) \rightarrow L^2_a(\Omega)$ is closed, densely defined, and self-adjoint. 
	\item 	The operator $A_a^b:\mathcal{D}(A_a^b) \rightarrow L^2_a(\Omega)$  has a purely discrete spectrum. 
	\item $-A^b_a$ generates a semigroup of operators $(\mathcal{T}^b_a(t))_{t \geq 0}$. The semigroup $(\mathcal{T}^{b}_a(t))_{t \geq 0}$ is positive and a contraction.
	\end{enumerate}

 \begin{proof}
 Consider the associated form $\omega_a^b$. This form is $\it{closed}$, i.e., the space $\mathcal{D} (\omega_a)$ equipped with the norm $\|\cdot\|_{\omega_a^b}$, given by 
$
\|u\|_{\omega_a^b} = (\|u\|^2_a+ \omega_a^b(u,u))^{1/2}
$
for each $u \in \mathcal{D}(\omega_a^b)$, is complete. This is true due to the fact that the multiplication map $u \mapsto  a \cdot u$ is an isomorphism from $WH^1_a(\Omega)$ to $WH^1(\Omega)$ and $WH^1(\Omega)$ is a Banach space. Moreover, the space $WH^1_a(\Omega)$ is dense in $L^2_a(\Omega)$. This follows from the inequality $\|au-av\|_2 \leq \|a\|_{\infty}\|u-v\|_2$ for each $u, v \in L^2(\Omega)$, the fact that the spaces $L^2_{\mathbf{1}}(\Omega)$ and $L^2_{a}(\Omega)$ are isomorphic, and the fact that $WH^1(\Omega)$ is dense in $L^2(\Omega)$. For the case when the domain is $\epsilon-\delta$ (Definition \ref{def:epsdel}) and $M= \mathbb{R}^n$, the density of the space $WH^1(\Omega)$ in $L^2(\Omega)$ has been established [Theorem 1.13]\cite{garofalo1996isoperimetric}. For the case when $M$ is compact, the density of the space follows trivially from the fact that the set of smooth functions on $M$ is dense in $L^2(\Omega)$. A similar result is also true for the case when the operator is elliptic, since $WH^1(\Omega)$ coincides with the usual Sobolev space [Theorem 2.9]\cite{aubin2013some}. In addition, it follows from the definition of the form $\omega_a^b$ that $\omega_a^b$ is \textit{symmetric}, meaning that $\omega_a^b(u,v) = \overline{\omega_a^b(v,u)}$ for each $u, v \in \mathcal{D}(\omega_a)$. The form $\omega_a^b$ is also \textit{semibounded}, i.e., there exists $m \in \mathbb{R}$ such that $\omega_a^b(u,u) \geq m \|u\|^2_a$ for each $u \in \mathcal{D}(\omega_a^b)$. In particular, this inequality is true for $m=0$ since $\omega_a^b(u,u)$ is non-negative for all $u \in \mathcal{D}(\omega_a^b)$. Hence, it follows from \cite{schmudgen2012unbounded}[Theorem 10.7] that the operator $A_a^b$ is self-adjoint.

To establish the discreteness of the spectrum of $A_a^b$, we need to establish that the space $WH^1_a(\Omega)$ is compactly embedded in $L^2_a(\Omega)$ whenever one of the conditions \ref{cond1}-\ref{cond3} in the statement of the proposition is true. This is true if  $WH^1(\Omega)$ is compactly embedded in the space $L^2(\Omega)$, which is known to be true when $\Omega$ is a subset of $\mathbb{R}^n$ and an $\epsilon$-$\delta$ domain [Theorem 1.27]\cite{garofalo1996isoperimetric} or when  ${\rm span} \{X_i(\mathbf{x}) \} = T_\mathbf{x}M$ [Theorem 2.33]\cite{aubin2013some}. For the case when $M$ is a compact manifold without a boundary, we note that $u \in \mathcal{D}(A^\mathbf{1}_\mathbf{1})$ implies that $u$ lies in some fractional Sobolev space $H^{r}(\Omega)$ with its fractional Sobolev norm uniformly bounded [Localization lemma]\cite{kohn1973pseudo}. Since the fractional Sobolev space $H^{r}(\Omega)$ is compactly embedded in $L^2(\Omega)$ whenever the exponent $r$ is positive, we can infer that $WH^1_a(\Omega)$ is compactly embedded in $L^2_a(\Omega)$. This implies that when $WH^1_a(\Omega) = \mathcal{D}(\omega_{a,b})$ is equipped with the norm $\|\cdot\|_{\omega_{a,b}}$, then it is  also compactly embedded in $L^2_a(\Omega)$. From \cite{schmudgen2012unbounded}[Proposition 10.6], it follows that $A_a^b$ has a purely discrete spectrum.
 \end{proof}
\end{lemma}
\kar{That the operator $A_a^b$ is densely defined follows from the fact that the domain of the form $\omega_a^b$, which is $WH_a^1(\Omega)$, is dense in $L^2_a(\Omega)$ when equipped with the norm $(\omega_a^b(u,u) + \|u\|^2_2)^{1/2}$.} The  result on discreteness of the spectrum of the operator $A^b_a$ follows from the compactness of the embedding $WH^1(\Omega)$ in $L^2(\Omega)$. The positivity properties of the semigroup $\mathcal{T}^b_{a}(t)$, we need the fact that if $f \in WH^{1}(\Omega)$, then $|f| \in WH^{1}(\Omega)$. This result is known for the case where $\Omega$ is a subset of $\mathbb{R}^n$. The proof for the general case where $\Omega$ is a manifold follows verbatim the results of \cite{garofalo1996isoperimetric,nhieu2001neumann}, since the proof only requires the density of the space $C^{\infty}(M)$ in $WH^1(\Omega)$, which can be verified. \kar{Using these facts, the proof of the previous lemma follows the proof of \cite{elamvazhuthi2018bilinear}[Lemma IV.1] very closely. The contractivity of the semigroup follows from \cite{ouhabaz2009analysis}[Proposition 1.51] due to the fact that the form $\omega_a^b$ is \textit{accretive}; that is, $ \omega_b^a(u,u) \geq 0$ for all $u \in \mathcal{D}(\omega_b^a)$.}

\begin{proposition}
\label{prop:garnh}
\cite{garofalo1996isoperimetric,nhieu2001neumann} Given $f \in WH^{1}(\Omega)$, we have that $|f| \in WH^{1}(\Omega)$.
\end{proposition}

Using the above proposition and the established properties of the operator $A_a^b$, we can prove the following results on the semigroup generated by the operator $-A^b_a$ concerning its positivity and mass-conserving properties. These results will be used in the analysis presented in Section \ref{sec:Disco}.

\begin{corollary}
	\label{CorIIhol}
 \kar{The operator $-A^b_a$ generates a semigroup of operators $(\mathcal{T}^b_a(t))_{t \geq 0}$ acting on $L^2_a({\Omega})$. }The semigroup $(\mathcal{T}^{b}_a(t))_{t \geq 0}$ is positive.
Furthermore, if $a=b = \mathbf{1}$ is the constant function equal to $1$ almost everywhere, then $\|y^0  \|_{\infty} \leq 1$ implies that $\| \mathcal{T}^{b}_a(t)y^0  \|_{\infty} \leq 1$ for all $t\geq 0$. 
	
	Additionally, this semigroup has the following {\it mass conservation property}: if $y_0 \in L^2_a(\Omega)$, $y^0 \geq 0$ and $\int_{\Omega}y^0(\mathbf{x})d \mathbf{x} =1$, then $\int_{\Omega} (\mathcal{T}_a^b(t)y^0)(\mathbf{x})d \mathbf{x} = \int_{\Omega} (\mathcal{T}^b_a(t)y^0)(\mathbf{x})d \mathbf{x} = 1 $ for all $t \geq 0$. 
\end{corollary}
\begin{proof}
  First, we note that the operator $-A_a^b$ is \textit{dissipative}, i.e., $\|(\lambda + A_a^b)u\|_a \geq \lambda \|u\|_a$ for all $\lambda>0$ and all $u \in \mathcal{D}(A_a^b)$, since $\omega_a(u,u) \geq 0$ for all  $u \in \mathcal{D}(\omega_a^b)$. Next, we note that $-A_a^b$ is self-adjoint, and hence the adjoint operator $-(A_a^b)^*$ is dissipative as well. It follows from a corollary of the {\it Lumer-Phillips theorem} \cite{engel2000one}[Corollary II.3.17] that $-A_a^b$ generates a semigroup of operators $(\mathcal{T}_a^b(t))_{t \geq 0}$ that solve the PDE \eqref{eq:Mainsysan} in the mild sense. 

 From Proposition \ref{prop:garnh}, we have that $v \in WH^1(\Omega)$ implies that $|v| \in WH^1(\Omega)$ whenever $v$ is only real-valued. This implies that if $u \in \mathcal{D}(\mathcal{\omega}_a)$, then $|{\rm Re}(u)| \in \mathcal{D}(\mathcal{\omega}_a^b)$, where ${\rm Re}(\cdot)$ denotes the real component of its argument. Then the  positivity of the semigroup follows from \cite{ouhabaz2009analysis}[Theorem 2.7]. 

To prove the last statement in the corollary, consider the closed convex set $C = \lbrace u \in L^2(\Omega); ~{\rm Re}( u) = u, ~ u(\mathbf{x}) \leq 1 ~ {a.e.} 	~\mathbf{x} \in \Omega \rbrace$. The projection of a function $u \in L^2_a(\Omega)$ onto the set $C$ can be represented by the (nonlinear) operator $P$, given by $P u = {\rm Re}(u) \wedge 1/a = \frac{1}{2}{\rm Re} (u) + \frac{1}{2}|{\rm Re} (u)-1|$. If $u \in \mathcal{D}(\omega_a)$, then it follows from the chain rule that $\nabla (Pu) = \frac{1}{2}{\rm sign}({\rm Re}(u) - 1 ) \nabla ({\rm Re}(u))+\frac{1}{2}\nabla ({\rm Re}(u))$. Hence, it follows that $\omega_\mathbf{1}^\mathbf{1}(Pu,Pu) \leq   \omega_\mathbf{1}^\mathbf{1}(u,u)$ for all $u \in \mathcal{D}(\omega_\mathbf{1}^\mathbf{1})$. According to \cite{ouhabaz2009analysis}[Theorem 2.3], this implies that the set $C$ is invariant under the positive semigroup  $(\mathcal{T}^\mathbf{1}_\mathbf{1}(t))_{t \geq 0}$;  therefore, we can conclude that if $\|y^0  \|_{\infty} \leq 1$, then $\|\mathcal{T}_\mathbf{1}^\mathbf{1}(t)y^0  \|_{\infty} \leq 1$ for all $t\geq 0$.

Lastly, we consider the mass conservation property. Let $\int_\Omega y^0(\mathbf{x}) d \mathbf{x} = 1 $ such that $y^0 \in L^2_a(\Omega)$. Then $\int_{\Omega}(y(\mathbf{x},t)-y^0( \mathbf{x}))d \mathbf{x} =  -\int_\Omega A^b_a(\int_0^t y(\mathbf{x},s) ds)d \mathbf{x}= -\omega^b_a(\int_0^t y(\mathbf{x},s) ds,1/a)=0$ for all $t \geq 0$. Hence, the integral preserving property of the semigroup holds. 
\end{proof}

Finally, we establish the following important result on the long-term stability properties of the semigroups associated with the operators in $\sum_{i =1}^m X^*_i (b(\mathbf{x}) X_i (a(\mathbf{x} ) ~ \cdot ~))$.

\begin{theorem}
\label{thm:stablin}
\textbf{(Exponential stability of semigroup)} \newline
    	The semigroup $(\mathcal{T}_a^b(t))_{t \geq 0}$ generated by the operator $-A_a^b$ is analytic.
Moreover, $0$ is a simple eigenvalue of 
$-A_a^b$ corresponding to the eigenvector $f=1/a$. 
	Hence, if $y^0 \geq 0$ and $\int_{\Omega}y^0(\mathbf{x})d \mathbf{x} = \int_{\Omega}f(\mathbf{x})d \mathbf{x} = 1$, then the following estimate holds  for some positive constants $M_0, \lambda$ and all $t \geq 0$:
	\begin{eqnarray}
	\|\mathcal{T}_a^b(t)y^0-f\|_{a} ~&\leq&~ M_0 e^{-\lambda t}\|y^0-f\|_a \label{eq:expcn11hol}   
	\end{eqnarray}
\end{theorem}
\begin{proof}
	 The operator $A_a^b$ is self-adjoint. Hence, its spectrum lies in $[0,\infty)$. From this, it follows that the corresponding semigroup generated by $-A_a^b$ is  an analytic semigroup \cite{lunardi2012analytic}[Chapter II].  
  
  Next, we prove the stability property. \kar{The semigroup is compact due to the compactness of 
 the embedding $\mathcal{D}(A_a^b) \subset WH^1(\Omega)$ in $L^2(\Omega)$. Since the semigroup is analytic and compact,} in order to establish \kar{its  stability properties,   it is sufficient to identify the eigenvectors associated with the eigenvalue $0$.} In the proof of the corresponding result in \cite{elamvazhuthi2018bilinear},  we used the Poincar\'e inequality to establish the uniqueness of the eigenvector of constant functions, corresponding to the eigenvalue $0$ of the Laplacian $\Delta$. It is not clear whether the Poincar\'e inequality holds for the operator $-A^b_a$ for condition \ref{cond2} in Assumption \ref{asmp1}. Hence, instead of using a Poincar\'e inequality, we will prove that the kernel of the operator $ -A^b_a$ consists only of constant functions. Suppose $u \in \mathcal{D}(A)$ is such that $Au = \mathbf{0}$, where $A:=A^\mathbf{1}_\mathbf{1}$. This implies that $<Au,u>$ $= \int_{\Omega}\sum_{i=1}^m(X_i u)^2d\x= 0$. Since the operator $A$ satisfies the Lie rank condition, from regularity results due to  H\"{o}rmander  \cite{bramanti2014invitation}, we can infer that $u$ is locally smooth everywhere in $\Omega$.
Then we know that, for a given horizontal curve $\ga:[0,1] \rightarrow \Omega$, 	\kar{$u(\ga(1)) = u(\ga(0)) +\int_{0}^1\sum_{i = 1}^ma_i(t)X_i u({\ga(t)})dt$ since $u(\ga(t))$ satisfies the differential equation $\dot{u}(\ga(t)) = \sum_{i = 1}^ma_i(t)X_iu(\ga(t))$, where $a_i(t)$ are the essentially bounded functions associated with the curve $\ga(t)$ according to  \eqref{eq:ctraffhyp}. Hence, $u(\ga(1)) -u(\ga(0)) =\int_{0}^1\sum_{i = 1}^ma_i(t)X_i u({\ga(t)})dt=0$ because $\sum_{i=1}^m(X_i u)^2= 0$.
}
	Note that we require the local smoothness of $u$ in order to make sense of the term $\int_{0}^1\sum_{i = 1}^ma_i(t)X u({\ga(t)})dt$. Since $\mathcal{V}$ is bracket generating, we can choose $\ga(t)$ such that $\ga(0)$ and $\ga(1)$ are the given initial and final conditions in $\Omega$. Hence, we have that $u$ is constant everywhere on $\Omega$. This implies that $A\mathbf{1} = \mathbf{0}$, and hence $A_a^bf= \mathbf{0}$ due to the assumption that $a,b $ are uniformly bounded from below by a positive constant. 
\end{proof}

\section{Stabilization with Local Agent Interactions}
\label{sec:Disco}

In Section \ref{sec:subell}, the probability densities that we stabilized were assumed to be uniformly bounded from below by a positive number.  Without this assumption, the semigroups that were constructed would not be globally asymptotically stable. In this section, we will introduce a semilinear PDE model for stabilizing a swarm to probability densities that possibly have supports that are disconnected.

As in Section \ref{sec:subell}, $\Omega$ will denote an open bounded subset of a manifold, and we consider  a collection of vector fields $\mathcal{V} = \lbrace  X_1,...,X_m \rbrace$ satisfying the Chow-Rashevsky condition. Let $A:=A^\mathbf{1}_\mathbf{1} $ be the operator defined in Section \ref{sec:subell}, where $\mathbf{1}$ denotes the function that is equal to $1$ almost everywhere on $\Omega$. We will also need the spaces $\mathbf{L}^2(\Omega) = L^2(\Omega) \times L^2(\Omega)$ and $\mathbf{L}^\infty(\Omega) = L^\infty(\Omega) \times L^\infty(\Omega)$ with the standard norms inherited from the spaces $L^2(\Omega)$ and $L^{\infty}(\Omega)$. 

We will consider the following PDE model, 
\begin{eqnarray}
&(y_1)_t = -A y_1- q_1(\mathbf{x},t)y_1 +q_2(\x,t)y_2  &~ in  ~~ \Omega \times [0,T] \nonumber \\ 
&(y_2)_t =           q_1(\x,t)y_1 -q_2(\x,t)y_2         & ~ in  ~~ \Omega \times [0,T] \nonumber \\
&\y(\cdot,0) = \y^0 & \hspace{-2cm} ~ in ~~ \Omega \nonumber \\ 
&\mathbf{n} \cdot \nabla y_1 =0 & \hspace{-2cm} 
 ~ in ~~ \partial \Omega \times [0,T], 
\label{eq:vecpdeini}
\end{eqnarray}
where $y_1$ and $y_2$ are non-negative functions and $q_i$ are reaction parameters.
This PDE model is the forward equation of a hybrid switching diffusion process (HSDP) \cite{yin2010hybrid}. In addition to a continuous spatial state $\mathbf{Z}(t)$, each agent is associated with a discrete state $Y(t) \in \{0,1\}$ at each time $t$. The hybrid switching diffusion process $(\mathbf{Z}(t),Y(t))$ can be represented as a system of SDEs of the form  
\begin{eqnarray}
	d\mathbf{Z}(t) &=&  \sqrt{2}(1-Y(t)) \sum_{i=1}^mX_i\circ dW_i \ +	\mathbf{n}(\mathbf{Z}(t))d\psi(t), \nonumber \\ 
	\mathbf{Z}(0) &=& ~ \mathbf{Z}_0.
\label{eq:litSDE2}
\end{eqnarray} 
The PDE \eqref{eq:vecpdeini} is related to the SDE \eqref{eq:litSDE2}, for each $k \in \{0,1\}$, through the relation $\mathbb{P}(Y(t)=k,\mathbf{Z}(t) \in \Gamma) = \int_{\Gamma} y_{\kar{k+1}}(\mathbf{x},t)d\mathbf{x}$ for all $t \in [0,T]$ and all measurable $\Gamma \subset \Omega$. The transitions of the variable $Y(t)$ from one discrete state to another is determined by two functions $q_i:\Omega \rightarrow [0,\infty]$ in the following way, 
\begin{eqnarray} 
\label{eq:ch2marko}
\mathbb{P}(Y(t+h) = 1 | Y(t) = 0) = q_1(\mathbf{Z}(t),t)h + o(h) \\
\mathbb{P}(Y(t+h) = 0| Y(t) = 1) = q_2(\mathbf{Z}(t),t)h + o(h) 
\end{eqnarray}
The state $Y(t) =0$  corresponds to the state in which agents diffuse in space according to the reflected SDE, and the state $Y(t)=1$ corresponds to a state in which they are motionless.
Therefore, unlike the process considered in Section \ref{sec:subell}, each agent has two discrete states, between which it jumps according to the {\it transition rates} $q_i(\mathbf{x},t)$ (also called {\it reaction parameters}). We will treat the transition rates $q_i(\mathbf{x},t)$  as the control inputs, instead of the velocity and diffusion parameters $(u_i,v_i)$. Since we will allow the control inputs to be functions of the density of the random variables $(\mathbf{Z}(t),Y(t))$, \kar{i.e., the density of agents in each state at time $t$,} this reaction-based control mechanism \kar{depends on  {\it interactions} among 
agents that enable them to estimate these  densities, e.g., via local sensing, wireless communication, or physical encounters. Due to the density-dependent transition rates, 
the forward equation is a semilinear PDE.}

We will consider the following problem in this section.
\begin{problem}
	\label{prb:dissppde}
	Let $y^d \in L^{\infty}(\Omega)$ be a target probability density. Construct a mean-field feedback law $K_i: \eL^2(\Omega) \rightarrow L^{\infty}(\Omega)$ such that if $u_i(\cdot,t)=K_i(\y(t))$ for all $i \in \lbrace 1,2 \rbrace$ and all $t \geq 0$, then the system \eqref{eq:vecpdeini} is globally asymptotically stable about the equilibrium $\y^d = [\mathbf{0}~~ y^d]^T$.
\end{problem}

Before we address this problem, we make some additional assumptions on the domain $\Omega$ and the operator $A$. Toward this end, we present the following definitions.

	\begin{definition}
		\label{def:c11}
	We will say that $\Omega$ is a \textbf{$C^{1,1}$ domain} if each point $\mathbf{x} \in \partial \Omega$ has a neighborhood $\mathcal{N}$ such that $\Omega \cap \mathcal{N}$ is represented by the inequality $x_n < \gamma(x_1,...,x_{n-1})$ in some Cartesian coordinate system for some function $\gamma : \mathbb{R}^{n-1} \rightarrow \mathbb{R}$ that is at least once differentiable and has derivatives of order $1$ that are Lipschitz continuous. 
\end{definition}

\begin{definition}
	\label{def:chain}
	The domain $\Omega$ will be said to satisfy the \textbf{chain condition} if there exists a constant $C>0$ such that for every $\mathbf{x},\bar{\mathbf{x}}$ $\in$ $\Omega$ and every positive $\kar{j} \in \mathbb{Z}_+$, there exists a sequence of points $\mathbf{x}_i$ $\in$ $\Omega$, $0 \leq i \leq \kar{j}$, such that $\mathbf{x}_0 = \mathbf{x}$, $\mathbf{x}_\kar{j} =\bar{\mathbf{x}}$, and $|\mathbf{x}_i-\mathbf{x}_{i+1}| \leq \frac{C}{\kar{j}} |\mathbf{x} - \bar{\mathbf{x}}|$ for all $i = 0,...,\kar{j}-1$. Here $|\cdot |$ denotes the standard Euclidean norm.
\end{definition}

Note that every convex domain satisfies the chain condition. 

In this section, we will make some stronger assumptions on the generator and the domain $\Omega$ than those made in the previous section. 
\begin{assumption}
\label{asmp2}
Only one of the following conditions holds:
\begin{enumerate}
	\item If $\Omega \neq M$, then $\Omega$ is a bounded subset of $\kar{\mathbb{R}^n}$, 

 $-A = \sum_{i=1}^{\kar{n}}\partial_{x_i}^2=\Delta$ is the Laplacian, and $\Omega$ is a $C^{1,1}$ domain in the sense of Definition \ref{def:c11} and satisfies the chain condition in Definition \ref{def:chain}.
	\item The set $\Omega$ is a compact manifold $M$ without a boundary \kar{and $-A = \sum_{i=1}^m X_i^*X_i$.}
\end{enumerate}
\end{assumption}

Given these assumptions, we have the following result due to Gaussian estimates proved by \cite{choulli2015gaussian} for the Laplacian $\Delta$, and by \cite{jerison1986estimates} for sub-Laplacians. We will use this result in the subsequent analysis. 

\begin{theorem}
	\label{thm:origbnd}
	Let $(\mathcal{T}(t))_{t \geq 0}$ be the semigroup generated by the operator $-A$. Let $y^0\in L^2(\Omega)$ be non-negative. Then there exists a constant $C>0$ and time $T >0$, independent of $y^0$, such that $\mathcal{T}(t)y^0 \geq C \|y^0\|_1 $ for all $t \geq T$.
\end{theorem}

In order to address Problem \ref{prb:dissppde}, we define the following maps $F_i: L^2(\Omega) \rightarrow \kar{L^\infty(\Omega)}$, $i \in \lbrace 1,2 \rbrace$, 
\begin{equation}
\label{eq:reacp}
(F_i(f))(\mathbf{x}) = r_i(f(\mathbf{x})-y^d(\mathbf{x}))
\end{equation}
for almost every $\mathbf{x} \in \Omega$ and all $f \in L^2(\Omega)$, where $\kar{r_i}:\mathbb{R} \rightarrow [0,\beta ]$ are globally Lipschitz functions for some positive number $\beta$, such that the functions $r_1$ and $r_2$ have supports equal to the intervals \kar{$(-\infty,0]$ and $[0,\infty)$}, respectively. Our candidate mean-field feedback law $K_i$ for addressing Problem \ref{prb:dissppde} will be $K_i(\y) = F_i(\kar{y_2})$ for each $i \in \lbrace 1,2 \rbrace$.
Then the resulting {\it closed-loop} PDE is given by 
\begin{eqnarray}
\label{eq:clpPDEdsicon}
&(y_1)_t = -A y_1 - F_1(y_2)y_1 +F_2(y_2)y_2  &~~ in  ~~ \Omega \times [0,T] \nonumber \\ 
&(y_2)_t =           F_1(y_2)y_1 - F_2(y_2)y_2         &~~ in  ~~ \Omega \times [0,T] \nonumber \\
&\y(\cdot,0) = \y^0 & \hspace{-2cm} ~ in ~~ \Omega \nonumber \\ 
&\mathbf{n} \cdot \nabla y_1=0 & \hspace{-2cm} ~ in ~~ \partial \Omega \times [0,T],  
\label{eq:vecpde}
\end{eqnarray}
where the Neumann boundary condition in the last equation is specified only for the case where the boundary $\partial \Omega $ is nonempty. Since the transition rates are a functions of the distribution of the random variable, the relation between the system of SDEs \eqref{eq:litSDE2} and PDEs \eqref{eq:clpPDEdsicon} is no longer straightforward. For the choice of control law  $F_i$, the SDE becomes a stochastic process of {\it Mckean-Vlasov} type \cite{mckean1967propagation,kolokoltsov2010nonlinear}, and further analysis is required to establish a rigorous connection between the two systems \eqref{eq:vecpdeini} and \eqref{eq:clpPDEdsicon}. Such an analysis is beyond the scope of this article, and is left for  future work. 

Our main goal in this section will be to establish the asymptotic stability of the PDE \eqref{eq:clpPDEdsicon} given Assumption \ref{asmp2}. Before we begin the stability analysis of the above PDE model, we point out that standard approaches to stability analysis, such as linearization-based approaches or Lyapunov functional arguments, are not immediately applicable, as we demonstrate in the following two remarks.
\begin{remark} 
	\label{rmk:noexp}
		\textbf{(Lack of exponential stability)} 
Consider the linearization of the PDE \eqref{eq:clpPDEdsicon} about the target equilibrium density $\y^d = [\mathbf{0}~~ y^d]^T$. It can be verified that the (Fr\'echet) derivative of the nonlinear operators $F_i$ about 
$\y^d$ is the $\mathbf{0}$ operator. Therefore, the linearization of the PDE about the equilibrium  $\y^d$ is:
\begin{eqnarray}
\label{eq:lineapde}
&(\tilde{y}_1)_t = -A \tilde{y}_1   &~~ in  ~~ \Omega \times [0,T] \nonumber \\ 
&(\tilde{y}_2)_t =          \mathbf{0}      &~~ in  ~~ \Omega \times [0,T] \nonumber 
\end{eqnarray}
\end{remark}
Clearly, this PDE is not exponentially stable since the spectrum of its generator, 
$
\mathbf{A}\y = 
[
Ay_1~ \mathbf{0}
]^T
$,
has an infinite number of eigenvalues at $0$. Hence, the PDE \eqref{eq:clpPDEdsicon} cannot be locally exponentially stable about the equilibrium $\y^d$.
\begin{remark}
		\label{rmk:noLaSa}
	\textbf{(Difficulty in using LaSalle's principle)} 
	Another standard approach to establish asymptotic stability of dynamical systems is LaSalle's invariance principle \cite{khalil2002nonlinear}. However, the application of LaSalle's invariance principle for stability analysis of infinite-dimensional dynamical systems, such as the PDE \eqref{eq:clpPDEdsicon}, requires that the trajectories of the system remain in a compact set for all time. The compactness of trajectories for solutions of parabolic PDEs is usually inferred from the regularizing  effect of the diffusion component of the dynamics. This is not straightforward to establish for solutions $\mathbf{y}$ of the PDE \eqref{eq:clpPDEdsicon} due to the fact that the diffusion operator $A$ acts only on the first state $y_1$, and therefore it cannot be  guaranteed that the state $y_2$ lies in a Sobolev space.
\end{remark}
	
Due to the technical issues pointed out in Remarks \ref{rmk:noexp} and \ref{rmk:noLaSa}, we will use an alternative approach to establish asymptotic stability of the PDE \eqref{eq:clpPDEdsicon} based on the monotonicity properties of the PDE. In order to perform stability analysis of the PDE \eqref{eq:clpPDEdsicon}, we will need a suitable notion of a solution. Toward this end, we use the following definition.

\begin{definition}
	Let $(\mathcal{T}(t))_{t \geq 0}$ be the semigroup generated by the operator $-A$. We will say that the PDE \eqref{eq:clpPDEdsicon} has a \textbf{local mild solution} if there exist $T >0$ and $\y \in C([0,T];\eL^2({\Omega}))$ such that 	\begin{eqnarray}
	y_1(\cdot,t) &=& \mathcal{T}(t)y^0_1  - \int_0^t \mathcal{T}(t-s)\Big (F_1(y_2(\cdot,s)) y_1(\cdot,s) \Big) ds \nonumber \\ 
	&~&  + \int_0^t \mathcal{T}(t-s)\Big (F_2(y_2(\cdot,s))y_2(\cdot,s) \Big) ds, \nonumber \\
	y_2(\cdot,t) &=&  y^0_2+\int_0^t F_1 (y_2(\cdot,s)) y_1(\cdot,s)  ds \nonumber \\
	&~&- \int_0^t F_2 (y_2(\cdot,s))y_2(\cdot,s) ds 
	\end{eqnarray}
	for all $t \in [0,T]$. 
	We will say that the PDE \eqref{eq:clpPDEdsicon} has a unique  \textbf{global solution} if it has a unique local mild solution for every $T>0$.
\end{definition}

To establish the existence of solutions of the PDE \eqref{eq:clpPDEdsicon}, we will need the operator  $\mathbf{A}:\mathcal{D}(\mathbf{A}):\rightarrow \eL^2(\Omega)$, defined as 
\begin{equation}
\mathbf{A}\y = 
\begin{bmatrix}
Ay_1 \\ \mathbf{0}
\end{bmatrix} 
\nonumber
\end{equation}
for all $\y \in \mathcal{D}(\mathbf{A}) = \mathcal{D}(A) \times L^2(\Omega)$.

Our next goal will be to construct global solutions of the PDE \eqref{eq:clpPDEdsicon}. First, we will show that the solutions of the PDE \eqref{eq:clpPDEdsicon} remain essentially bounded if the initial condition is essentially bounded. Toward this end, we first establish this property for a related autonomous linear PDE.
\begin{lemma}
	\label{lem:Linffbnd}
	Suppose $\y \in \eL^{\infty}(\Omega)$. Let $\mathbf{a} \in \mathbf{L}^{\infty}(\Omega)$ be non-negative. Consider the linear bounded operator $\mathbf{B}:\mathbf{L}^2(\Omega)  \rightarrow \mathbf{L}^2(\Omega)$ defined by
	\begin{eqnarray}
	(\kar{\mathbf{B}}\mathbf{y})(\mathbf{x}) =
	\begin{bmatrix}
	-a_1(\mathbf{x})y_1(\mathbf{x})+a_2(\mathbf{x})y_2(\mathbf{x})  \nonumber \\ 
	a_1(\mathbf{x})y_1(\mathbf{x})-a_2(\mathbf{x})y_2(\mathbf{x})
	\end{bmatrix} \nonumber
	\end{eqnarray}
	for almost every $\mathbf{x} \in \Omega$ and all $\y \in \eL^2(\Omega)$. Let $(\mathcal{T}^\mathbf{C}(t))_{t \geq 0}$ be the semigroup generated by the operator $\mathbf{C} =-\mathbf{A}+\mathbf{B}$. Then $\|\mathcal{T}^\mathbf{C}(t)\mathbf{y}^0\|_{\infty} \leq e^{\|\mathbf{a}\|_{\infty}t}\|\mathbf{y}^0\|_{\infty}$ for all $t \geq 0$.
\end{lemma}
\begin{proof}
	We know that the operator $\mathbf{A}$ generates a semigroup $(\mathcal{T}^\mathbf{A}(t))_{t \geq 0}$ given by 
	\begin{equation}
	\mathcal{T}^\mathbf{A}(t) = 
	\begin{bmatrix}  
	\mathcal{T}(t)  & \mathbf{0} \\
	\mathbf{0}   & \mathbf{I}
	\end{bmatrix}
	\end{equation}
	for all $t \geq 0$. Moreover, the semigroup $(\mathcal{T}^\mathbf{A}(t))_{t \geq 0}$ satisfies $\|\mathcal{T}^\mathbf{A}(t) \y^0\|_{\infty} \leq \|\y^0\|_{\infty}$ for all $\y^0\in \eL^{\infty}(\Omega)$ and $t \geq 0$ (Corollary \ref{CorIIhol}). Additionally, we know that the semigroup  $(\mathcal{T}^\mathbf{B}(t))_{t \geq 0}$ generated by the bounded operator $\mathbf{B}$ satisfies the estimate $ \|\mathcal{T}^\mathbf{B}(t)\y^0\|_{\infty}   \leq e^{\|\mathbf{a}\|_{\infty}t}\|\mathbf{y}^0\|_{\infty}$. \kar{Since $\mathbf{B}$ is a bounded operator, and the resolvent of $\mathbf{C}$ has an explicit well-defined representation \cite{engel2000one}[p.160], we can conclude that $\lambda -\mathbf{C}$ has a dense range in $\mathbf{L}^2(\Omega)$ for all $\lambda$.} Then the result follows from the {\it Lie-Trotter product formula} \cite{engel2000one}[Corollary III.5.8] 
 by noting that $\mathcal{T}^\mathbf{C}(t) =   \lim_{N \rightarrow 0} (\mathcal{T}^\mathbf{A}(\frac{t}{N})\mathcal{T}^\mathbf{B}(\frac{t}{N}))^N $, where the limit holds in the strong operator topology, for all $t \geq 0$.
\end{proof}

Now we can show that the $\eL^\infty-$ estimate proved in Lemma \ref{lem:Linffbnd} can be extended to a class of non-autonomous linear systems that can be treated as autonomous linear systems over certain intervals of time.

\begin{lemma}
	Suppose $\y^0 \in \eL^{\infty}(\Omega)$, $c>0$ and $T>0$. \sloppy Let $a_1,a_2 \in L^{2}(0,T;L^2(\Omega))$ be non-negative and piecewise constant with respect to time, with $\|a_1(t)\|_{\infty} \leq c $ and $\|a_2(t)\|_{\infty} \leq c $ for all $t \in [0,T]$. Then suppose $\y \in C([0,T];\eL^2(\Omega))$ is given by 
	\begin{eqnarray}
	y_1(\cdot,t) &=& \mathcal{T}(t)y^0_1 - \int_0^t \mathcal{T}(t-s)\bigg ( a_1(\cdot,t) y_1(\cdot,s) \bigg)ds \nonumber  \\ 
	&~& + \int_0^t \mathcal{T}(t-s) \bigg ( a_2(\cdot,t) y_2(\cdot,s) \bigg) ds \nonumber \\
	y_2(\cdot,t) &=&  y^0_2+\int_0^t a_1(\cdot,s) y_1(\cdot,s) ds  - \int_0^t a_2(\cdot,s) y_2(\cdot,s) ds \nonumber
	\end{eqnarray}
	for all $t \in [0,T]$. Then
	\begin{equation}
	\kar{\|\y(\cdot,t)\|_{\infty}} \leq e^{ct}\|\mathbf{y}^0\|_{\infty}
	\label{eq:linfgro}
	\end{equation}
	for all $t \in [0,T]$.
\end{lemma}
\begin{proof}
	Let $(t_i)_{i = 0}^{m}$ be a finite sequence of length $m+1 \in \mathbb{Z}_+$ of strictly increasing time instants, with $t_0=0$, such that the functions $a_1$ and $a_2$ are constant over the intervals $[t_{i-1},t_{i})$, $i \in \lbrace 1,...,m \rbrace$. Then, for each $i \in \lbrace 1,...,m \rbrace$, consider the bounded operators $\mathbf{B}_i: \eL^2(\Omega) \rightarrow \eL^2(\Omega)$ and $\mathbf{C}_i: \mathcal{D}(\mathbf{A}) \rightarrow \eL^2(\Omega)$ 
	\begin{eqnarray}
	(\mathbf{B}_i \y)(\mathbf{x}) =
	\begin{bmatrix} 
	-a_1(\mathbf{x},t_{i-1})y_1(\mathbf{x})+a_2(\mathbf{x},t_{i-1})y_2(\mathbf{x})  \\ 
	a_1(\mathbf{x},t_{i-1})y_1(\mathbf{x})-a_2(\mathbf{x},t_{i-1})y_2(\mathbf{x})
	\end{bmatrix} 
	\end{eqnarray}
	for almost every $\x \in \Omega$ and all $\y \in \eL^2(\Omega)$, and $\mathbf{C}_i = \mathbf{A}+\mathbf{B}_i$, respectively. Then for each $i \in \lbrace 1,...,m \rbrace$, $\y $ is given by
	\begin{equation}
	\y(\cdot,t) = \mathcal{T}^{\mathbf{C}_i}(t - t_i)\mathcal{T}^{\mathbf{C}_{i-1}}(t_{i}-t_{i-1})... \mathcal{T}^{\mathbf{C}_1}(t_1)
	\end{equation}
	for all $ t \in [t_{i-1}, t_{i}] $. Then the result follows from Lemma \ref{lem:Linffbnd}.
\end{proof}
\begin{lemma}
	\label{lem:linfbnd}
	Suppose $\y^0 \in \eL^{\infty}(\Omega)$, $c>0$, and $T>0$.  \sloppy Let the functions $a_1,a_2 \in L^{2}(0,T;L^2(\Omega))$ be non-negative with $\|a_1(t)\|_{\infty} \leq c $ and $\|a_2(t)\|_{\infty} \leq c $ for almost every $t \in [0,T]$,  Then \kar{there exists $\y \in C([0,T];\eL^2(\Omega))$ given by} 
	\begin{eqnarray}
	y_1(\cdot,t) &=& \mathcal{T}(t)y^0_1 - \int_0^t \mathcal{T}(t-s)\Big (a_1(\cdot,s) y_1(\cdot,s) \Big ) ds  \nonumber  \\ 
	&~& + \int_0^t \mathcal{T}(t-s) \Big ( a_2(\cdot,s) y_2(\cdot,s) \Big ) ds, \nonumber \\
	y_2(\cdot,t) &=&  y^0_2+\int_0^t a_1(\cdot,s) y_1(\cdot,s) ds  - \int_0^t a_2(\cdot,s) y_2(\cdot,s) ds \nonumber \\
	\label{eq:vareq1}
	\end{eqnarray}
	for all $t \geq 0$.  \kar{Moreover,}
	\begin{equation}
	\|\y(\cdot,t)\|_{\infty} \leq e^{ct}\|\mathbf{y}^0\|_{\infty}
	\label{eq:linfgrores}
	\end{equation}
	for all $t \in [0,T]$.
\end{lemma}
\begin{proof}
	Given that $a_1,a_2 \in L^{2}(0,T;L^2(\Omega))$, we know that there exists \sloppy a sequence of \kar{piecewise constant} (with respect to time) non-negative functions  $(a_1^i)_{i=1}^{\infty},(a_2^i)_{i=1}^{\infty}$ in $L^{2}(0,T;L^2(\Omega))$ such that $
	\lim_{i \rightarrow \infty}\|a_j^i-a_j\|_{L^{2}(0,T;L^2(\Omega))}=0$, for $j =1,2$ \cite{roubivcek2013nonlinear}[Proposition 1.36]. Moreover, for each $j \in \lbrace 1,2\rbrace$, we can assume that $\|a_j^i(t)\|_{\infty} \leq c$ for all $t \in [0,T]$ and all $i \in \mathbb{Z}_+$. Consider the corresponding sequence $(\y)_{i=1}^\infty$ in $C([0,T];\eL^2(\Omega))$ defined by 
	\begin{eqnarray}
 \label{eq:nonauto}
	y^i_1(\cdot,t) \hspace{-2mm} &=& \hspace{-1mm} \mathcal{T}(t)y^0_1 - \int_0^t \mathcal{T}(t-s)\Big ( a^i_1(\cdot,s) y^i_1(\cdot,s) \Big)ds \nonumber \\ &~& \hspace{-1mm} + \int_0^t \mathcal{T}(t-s)\Big ( a_2(\cdot,s) y^i_2(\cdot,s) \Big) ds, \nonumber \\
	y^i_2(\cdot,t) \hspace{-2mm} &=&  \hspace{-1mm} y^0_2+\int_0^t a^i_1(\cdot,s) y^i_1(\cdot,s) ds   - \int_0^t a^i_2(\cdot,s) y^i_2(\cdot,s) ds \nonumber \\
	\label{eq:vareqseq}
	\end{eqnarray}
	for each $i\in \mathbb{Z}_+$.  Let  $\kar{\mathbf{e}^{i,j}} \in C([0,T];\eL^2(\Omega))$ be given by \kar{$\mathbf{e}^{i,j}  =\y^i - \y^j$} for each \kar{$i,j \in \mathbb{Z}_+$}. Then, from equations \eqref{eq:vareq1} and \eqref{eq:vareqseq}, we know that \kar{$\mathbf{e}^{i,j}$} satisfies 
	\begin{eqnarray}
	\kar{e^{i,j}_1}(\cdot,t) &=&  - \int_0^t \mathcal{T}(t-s) \Big ( a^i_1(\cdot,s) y^i_1(s,\cdot) \Big )ds \nonumber \\
	&~&+ \int_0^t \mathcal{T}(t-s)\Big ( a^i_2(\cdot,s) y^i_2(s,\cdot) \Big) ds \nonumber \\
	&~& + \int_0^t \mathcal{T}(t-s)\Big ( a^{\kar{j}}_1(\cdot,s) y^{\kar{j}}_1(\cdot,s) \Big)ds \nonumber \\
	&~& - \int_0^t \mathcal{T}(t-s)\Big (a^{\kar{j}}_2(\cdot,s) y^{\kar{j}}_2(\cdot,s) \Big ) ds  \nonumber \\
	&=&  - \int_0^t \mathcal{T}(t-s)\Big ( (a^i_1(\cdot,s) -a^{\kar{j}}_1(\cdot,s)) y^i_1(\cdot,s)  \Big )ds  \nonumber \\
	&~ & + \int_0^t \mathcal{T}(t-s)\Big ( a^{\kar{j}}_1(\cdot,s) (y^{\kar{j}}_1(\cdot,s)-y^i_1(s,\cdot))  \Big )ds \nonumber \\
	& ~& + \int_0^t \mathcal{T}(t-s)\Big ( (a^i_2(\cdot,s)-a^{\kar{j}}_2(\cdot,t)) y^i_2(\cdot,s)  \Big )ds \nonumber \\
	&~ & - \int_0^t \mathcal{T}(t-s)\Big ( a^{\kar{j}}_2(\cdot,s) (y^{\kar{j}}_2(\cdot,s)-y^i_2(\cdot,s))  \Big )ds  \nonumber
	\end{eqnarray}
	for all $t \in [0,T]$.  Considering the fact that the semigroup $\mathcal{T}(t)$ is contractive (Proposition  \ref{opprop1hol}), and that $a_{\kar{k}}^i$ and $y^i_{\kar{k}}$ are uniformly bounded in $L^{\infty}((0,T) \times \Omega)$ \kar{for $k=1,2$,} we can conclude that there exists a constant $\alpha>0$ such that \kar{
	\begin{eqnarray}
	\|e^{i,j}_1(\cdot,t)\|_2 &\leq &  \alpha\|a_1^i-a^j_1\|_{\infty}  \|y^0_1\|_{\infty}  \nonumber \\
	&~& + ~\alpha\|a^j_1\|_{\infty} \int_0^t \|e^{i,j}_1(s)\|_2ds\nonumber \\
	&~& +  ~\alpha\|a_2^i-a^j_2\|_{\infty}  \|y^0_2\|_{\infty}    \nonumber \\
	&~& + ~\alpha \|a^j_2\|_{\infty} \int_0^t\|e^{i,j}_2(s)\|_{2} ds
	\label{eq:ineq1rcon}
	\end{eqnarray}
	for all $t \in [0,T]$. Similarly, we can obtain  the estimate
	\begin{eqnarray}
	\|e^{i,j}_2(\cdot,t)\|_2 &\leq& \alpha\|a_1^i-a^j_1\|_{\infty}  \|y^0_1\|_{\infty}  \nonumber \\
	&~& +~ \alpha\|a^j_1\|_{\infty}  \int_0^t\|e^{i,j}_1(s)\|_{2}ds  \nonumber \\
	&~& +~ \alpha \|a_2^i-a^j_2\|_{\infty}  \|y^0_2\|_{\infty}  \nonumber \\
	&~& +~ \alpha \|a^j_2\|_{\infty} \int_0^t\|e^{i,j}_2(s)\|_{2} ds
	\label{eq:ineq2rcon}
	\end{eqnarray}
	for all $t \in [0,T]$. 
	
	Then, by considering the sum $\|e^{i,j}_1(\cdot,t)\|_2+\|e^{i,j}_2(\cdot,t)\|_2 $, combining the two inequalities \eqref{eq:ineq1rcon} and \eqref{eq:ineq2rcon}, and applying the integral form of Gronwall's inequality \cite{evans1998partial}, we have that 
	\begin{equation}
\|e^{i,j}_1(\cdot,t)\|_2+\|e^{i,j}_2(\cdot,t)\|_2 ~\leq~  C_1e^{C_2t}
	\label{eq:gronres1}
	\end{equation}
	for all $t \in [0,T]$, where $C_1$ and $C_2$ are constants depending only on $\beta$ (the upper bound on the functions $r_i$) and $\|\mathbf{y}^0\|_{\infty}$. From the inequality \eqref{eq:gronres1}, we can infer that 
	\begin{equation*}
	\lim_{i,j \rightarrow \infty}\|\mathbf{e}^{i,j}\|_{C([0,T];\eL^2(\Omega))} = 0
	\end{equation*}
 
	This implies that the solution of 
 \eqref{eq:vareq1} is continuous with respect to $a_i$, and by continuity, one can construct a solution $\mathbf{y}$ corresponding to the coefficients $a_i$.} Considering the estimate \eqref{eq:linfgro}, we can conclude that  $\y$ satisfies the estimate \eqref{eq:linfgrores}.
\end{proof}

From the above lemma, we can conclude the following theorem on global existence of solutions of the PDE \eqref{eq:clpPDEdsicon}. 

\begin{theorem}
	Suppose $\y^0 \in \eL^{\infty}(\Omega)$. Then the PDE \eqref{eq:clpPDEdsicon} has a unique global mild solution. 
\end{theorem}
\begin{proof}
We use a contraction mapping approach to construct the solution. Consider a  map  $\Gamma$ on $C([0,T];\mathbf{L}^2(\Omega))$ that is defined by $\mathbf{v} \mapsto \Gamma (\mathbf{v}) \equiv \tilde{\mathbf{v}}$, where $\tilde{\mathbf{v}}$ is constructed by setting $a_i = F_i(v_2)$, $i \in \{1,2\}$, in \eqref{eq:vareq1}: 
\begin{eqnarray}
	\tilde{v}_1(\cdot,t) &=& \mathcal{T}(t)y^0_1  - \int_0^t\mathcal{T}(t-s)\Big (F_1(v_2(\cdot,s)) \tilde{v}_1(\cdot,s) \Big) ds \nonumber \\ 
	&~&  + ~ \int_0^t \mathcal{T}(t-s)\Big (F_2( v_2(\cdot,s))\tilde{v}_2(\cdot,s) \Big) ds, \nonumber \\
	\tilde{v}_2(\cdot,t)  &=&  y^0_2+\int_0^t F_1(v_2(\cdot,s)) \tilde{v}_1(\cdot,s)  ds \nonumber \\
	&~&- \int_0^t F_2 (v_2(\cdot,s))\tilde{v}_2(\cdot,s) ds 
 \label{eq:eq12}
	\end{eqnarray}
for all $t \in [0,T]$. 
 Then $\mathbf{y}$ is a solution of the PDE \eqref{eq:clpPDEdsicon} if it is a fixed point of the map $\Gamma$. From Lemma \ref{lem:linfbnd}, any solution of the PDE \eqref{eq:clpPDEdsicon} must be attained in the set  $\mathcal{Y} = \{  \mathbf{v} \in C([0,T],\mathbf{L}^2(\Omega));   \|v(t)\|_{\infty} \leq \hat{C} ~\forall t \in [0,T] \}$  for some sufficiently large constant $\hat{C} >0$ that depends on  $\|\mathbf{y}^0\|_{\infty}$. Therefore, it suffices to check that the map $\Gamma$ is a contraction on $\mathcal{Y}$, which guarantees that it has a unique fixed point. 
 
 Let $\mathbf{v},\mathbf{u} \in \mathcal{Y}$. To show that $\Gamma$ is a contraction on $\mathcal{Y}$, we will derive an upper bound on 
 $\sup_{t \in [0,T]}\|\Gamma(\mathbf{v})(t)-\Gamma(\mathbf{u})(t)\|_2$.  In order to estimate $\|\tilde{v}_1(\cdot,t) -\tilde{u}_1(\cdot,t)\|_2$, we first compute an upper bound on the corresponding terms in the difference according to the computation in \eqref{eq:eq12}:
\begin{align*}
& \hspace{-1mm} \bigg \|- \int_0^t\mathcal{T}(t-s)\Big (F_1(v_2(\cdot,s)) \tilde{v}_1(\cdot,s) \Big) ds \\ 
 &  ~~~~~ +  \int_0^t\mathcal{T}(t-s)\Big (F_1(u_2(\cdot,s)) \tilde{u}_1(\cdot,s) \Big) ds  \bigg \|_2~~~ \leq \\
& \hspace{-1mm} \bigg \| \int_0^t\mathcal{T}(t-s)\big (-F_1(v_2(\cdot,s)) \tilde{v}_1(\cdot,s) \\ 
& ~~~~~ + ~ F_1(u_2(\cdot,s)) \tilde{v}_1(\cdot,s) \big) ds \bigg \|_2 ~~~ + \\
& \hspace{-1mm} \bigg \| \int_0^t\mathcal{T}(t-s)\big (F_1(u_2(\cdot,s)) \tilde{v}_1(\cdot,s)  -  F_1(u_2(\cdot,s)) \tilde{u}_1(\cdot,s) \big) ds \bigg \|_2
\end{align*}
Since $r_i$ in the definition \eqref{eq:reacp} of $F_i$ is a bounded Lipschitz function (with, say, Lipschitz constant $K>0$), and $\mathcal{T}(t)$ is a contraction, we can compute the following upper bounds on the left-hand side of the above inequality:
\begin{align*}
& \bigg |\int_0^t\|\big (-F_1(v_2(\cdot,s))   +F_1(u_2(\cdot,s)\Big)  \tilde{v}_1(\cdot,s) \|_2 ds\\
& ~~~~~+~ c\int_0^t\|\tilde{v}_1(\cdot,s) -  \tilde{u}_1(\cdot,s) \|_2 ds \bigg | \\
& \leq ~~ K\int_0^t\|\Big (-v_2(\cdot,s)  +u_2(\cdot,s)\Big)  \tilde{v}_1(\cdot,s)\|_2  ds\\
& ~~~~~+~  c\int_0^t\|\tilde{v}_1(\cdot,s) -  \tilde{u}_1(\cdot,s) \|_2  ds \\
& \leq ~~ \hat{C} K\int_0^t \|v_2(\cdot,s)  - u_2(\cdot,s)\|_2  ds\\
& ~~~~~+~ c\int_0^t\|\tilde{v}_1(\cdot,s) -  \tilde{u}_1(\cdot,s) \|_2  ds \\
& \leq ~~ \hat{C}T K \sup_{t \in [0,T]}\|v_2(\cdot,t)  - u_2(\cdot,t)\|_2  \\
& ~~~~~+~ T  c \sup_{t \in [0,T]}\|\tilde{v}_1(\cdot,t) -  \tilde{u}_1(\cdot,t) \|_2.
\end{align*}
Using a similar computation, we can estimate that 
\begin{align*}
&\sup_{t \in [0,T]}(\|\tilde{v}_1(\cdot,t) -\tilde{u}_1(\cdot,t)\|_2 +\|\tilde{v}_2(\cdot, t)-\tilde{u}_2(\cdot,t)\|_2 ) \\ 
& \hspace{1cm} \leq ~~~~~ \tilde{C}T \sup_{t \in [0,T]}\|v_1(\cdot,t)  - u_1(\cdot,t)\|_2 \\
&  \hspace{1.5cm} + ~\tilde{C}T \sup_{t \in [0,T]}\|v_2(\cdot,t)  - u_2(\cdot,t)\|_2 \\
&  \hspace{1.5cm} +~ \tilde{C}T \sup_{t \in [0,T]}\|\tilde{v}_1(\cdot,t) -  \tilde{u}_1(\cdot,t) \|_2 \\
&  \hspace{1.5cm} +~ \tilde{C}T \sup_{t \in [0,T]}\|\tilde{v}_2(\cdot,t) -  \tilde{u}_2(\cdot,t) \|_2
\end{align*}
for a sufficiently large 
constant $\tilde{C}>0$ that depends only on $K$, $c$, and $\hat{C} $. This implies that 
\begin{align*}
& \sup_{t \in [0,T]}(\|\tilde{v}_1(\cdot,t) -\tilde{u}_1(\cdot,t)\|_2 +\|\tilde{v}_2(\cdot, t)-\tilde{u}_2(\cdot,t)\|_2 ) \\ 
& \hspace{0.75cm} \leq ~~~~\frac{\tilde{C}  T}{1-\tilde{C}T} \sup_{t \in [0,T]}\|v_1(\cdot,t)  - u_1(\cdot,t)\|_2 \\
& \hspace{1cm} ~ + ~\frac{\tilde{C}T}{1-\tilde{C}T} \sup_{t \in [0,T]}\|v_2(\cdot,t)  - u_2(\cdot,t)\|_2. 
\end{align*}
Thus, if $T>0$ is small enough, $\Gamma$ is a contraction map on $\mathcal{Y}$. This implies that for small enough $T>0$, $\Gamma$ has a fixed point $\mathbf{y}$ that is a unique local mild solution to the PDE \eqref{eq:clpPDEdsicon}. Then, due to the uniform bound on the solution established in Lemma \ref{lem:linfbnd}, it follows that the solution  can, in fact, be extended to arbitrary $T>0$, and is therefore global.
\end{proof}

Next, our goal will be to prove that $\y^d$ is the globally asymptotically stable equilibrium of the system \eqref{eq:clpPDEdsicon}. Toward this end, we first prove the following preliminary results.  

\begin{lemma}
	\label{lem:possolvar}
 Suppose $y^0 \in L^{\infty}(\Omega)$ and $T>0$. Let $a \in L^{\infty}(\Omega)$ be non-negative. Consider the multiplication operator $B:L^2(\Omega)  \rightarrow L^2(\Omega)$ defined by
	\begin{eqnarray}
	(By)(\mathbf{x}) =  -a(\mathbf{x})y(\mathbf{x}) \nonumber 
	\end{eqnarray}
	for all $\mathbf{x} \in \Omega$ and all $y \in L^{2}(\Omega)$. Let $(\mathcal{T}^C(t))_{t \geq 0}$ be the semigroup generated by the operator $C =-A+B$. Then $\|\mathcal{T}^C(t)y^0\|_{\infty} \leq \|y^0\|_{\infty}$ for all $t \geq 0$.
\end{lemma}
\begin{proof}
	We know that if $(\mathcal{T}(t))_{t \geq 0}$ is the semigroup generated by the operator $-A$, then from Corollary \ref{CorIIhol}, $\|\kar{\mathcal{T}(t)}\mathbf{y}^0\|_{\infty} \leq \|y^0\|_{\infty}$ for all $t \geq 0$. Moreover, $B$ generates the multiplication  semigroup $(e^{-a(\cdot)t})_{t\geq 0}$. Then the result follows from the Lie-Trotter formula \cite{engel2000one}.
\end{proof}

\begin{lemma}
	\label{lem:postimvar}
	Let $T>0$. Let $f,a \in L^2(0, T;L^{2}(\Omega))$ be non-negative functions defined such that $\|f(t)\|_{\infty}$ and $\|a(t)\|_{\infty}$ are bounded by a constant $C>0$ almost everywhere on $t \in [0,T]$. Suppose $e \in C([0,T];L^2(\Omega))$ is given by
	\begin{eqnarray*}
	e(\cdot,t) 	= &   - \int_0^t \mathcal{T}(t-s)\Big ( a(\cdot,s)e(\cdot,s) \Big) ds \\ 
 & + \int_0^t \mathcal{T}(t-s)f(\cdot,s) ds 
	\end{eqnarray*}
for all $t \in [0,T]$. Then $e(\cdot,t)$ is non-negative for all $t \in [0,T]$.
\end{lemma}
\begin{proof}
		Since the proof follows a similar line of argument as  the proof of Lemma \ref{lem:linfbnd}, we only sketch an outline of the proof here. As in the proof of Lemma \ref{lem:linfbnd}, for a given $a \in L^{2}(0,T;L^{2}(\Omega))$ we can construct a sequence $(a^i)_{i=1}^{\infty}$ in $L^{2}(0,T;L^{2}(\Omega))$ that is piecewise constant in time and converges in $L^{2}(0,T;L^{2}(\Omega))$, with $\|a^i(t)\|_\infty$ bounded almost everywhere on $[0,T]$ by $C>0$. Let the sequence $(e_i)_{i=1}^{\infty}$ in $C([0,T];L^2(\Omega))$ be given by
	\begin{eqnarray}
	e_i(\cdot,t) 	= & - \int_0^t \mathcal{T}(t-s)\Big ( a^i(\cdot,s)e_i(\cdot,s) \Big )ds \nonumber \\ & +~ \int_0^t \mathcal{T}(t-s)f(\cdot,s) ds \nonumber
	\end{eqnarray}
	for all $t \in [0,T]$. \kar{Since $e_i(t)$ is also the solution of the PDE $\dot{e}_i(t) = -A e_i(t) - a^i(\cdot,t) e_i(t) +f(\cdot,t)$ with initial condition equal to $0$,  and this solution can be constructed using the positive semigroup $\mathcal{T}^C(t)$ in Lemma \ref{lem:possolvar},} we can conclude that $(e_i)_{i=1}^{\infty}$ is non-negative for each $i \in \mathbb{Z}_+$. Then, using the fact that the sequences $(e_i)_{i=1}^{\infty}$ and $(a^i )_{i=1}^{\infty}$ are uniformly bounded in the spaces $C([0,T];L^2(\Omega))$ \kar{and} $L^{2}(0,T;L^{2}(\Omega))$, respectively, and applying Gronwall's lemma, the result follows.
\end{proof}

We can use Lemma \ref{lem:postimvar} prove the 
next result,  which will enable us to show 
later on that the rate of convergence of the solution $\y$ of the PDE \eqref{eq:clpPDEdsicon} toward $\mathbf{0}$ can be controlled by the rate of convergence of the solution of a related linear PDE.

\begin{theorem}
	\label{thm:cmppi}
	\textbf{(Comparison Principle)} \\
Let $T>0$.	Let $y^0 \in L^2(\Omega)$ and $f,g \in L^2(0, T;L^{2}(\Omega))$ be non-negative such that $\|f(t)\|_{\infty}$ and $\|g(t)\|_{\infty}$ are bounded by a constant $C_1>0$ almost everywhere on $t \in [0,T]$. Define the operator $C = -{A} - \|g\|_{\infty} \mathbf{I}$. Let $y(\cdot, t)$ be given by
\begin{eqnarray}
\label{eq:yEq}
y(\cdot,t) = &\mathcal{T}(t)y^0- \int_0^t \mathcal{T}(t-s)\Big ( g(\cdot,s)y(\cdot,s)  \Big)ds \nonumber \\
	~&+ \int_0^t \mathcal{T}(t-s)f(\cdot,s) ds 
\end{eqnarray}
for all $t \in [0,T]$. Then $y(\cdot,t) \geq  \mathcal{T}^C(t)y^0$ for all $t \in [0,T] $, where $(\mathcal{T}^C(t))_{t \geq 0}$ is the semigroup generated by $C$. 
\end{theorem}
\begin{proof}
	Let $\tilde{y}(\cdot,t) = \mathcal{T}^C(t)y^0$ for all $t \geq 0$. Then, we know that $\tilde{y}(\cdot,t)$ satisfies the equation
	\begin{equation}
 \label{eq:ytildeEq}
	\tilde{y}(\cdot,t) = \mathcal{T}(t)y^0 - \int_0^t \mathcal{T}(t-s)\|g\|_{\infty}\tilde{y}(\cdot,s) ds
	\end{equation}
	for all $t \in [0,T]$.
	Let $e = y - \tilde{y} $. \kar{Substituting equations \eqref{eq:yEq} and \eqref{eq:ytildeEq} for $y$ and $\tilde{y}$, respectively, we have that
 \begin{eqnarray}
	e(\cdot,t) &=&  - \int_0^t \mathcal{T}(t-s) \Big ( g(\cdot,s)y(\cdot,s) \Big )ds
	 \nonumber \\ 
	 &~&+ \int_0^t \mathcal{T}(t-s)f(\cdot,s) ds   \nonumber \\
	&~&+ \int_0^t \mathcal{T}(t-s)\|g\|_{\infty}\tilde{y}(\cdot,s) ds \nonumber
 \end{eqnarray}
 for all $t \in [0,T]$. By adding and subtracting the term $ \int_0^t \mathcal{T}(t-s) \Big ( g(\cdot,s)\tilde{y}(\cdot,s) \Big )ds$ in this expression for $e$, we obtain
	\begin{eqnarray}
	e(\cdot,t) 
	&=&  - \int_0^t \mathcal{T}(t-s)\Big ((g(\cdot,s) )e(\cdot,s) \Big ) ds  \nonumber \\
	&~&+ \int_0^t \mathcal{T}(t-s)f(\cdot,s) ds \nonumber \\
		&~& + \int_0^t \mathcal{T}(t-s)\Big ( (\|g\|_{\infty}-g(\cdot,s))\tilde{y}(\cdot,s) \Big ) ds  \nonumber 
	\end{eqnarray}
	for all $t \in [0,T]$.} Then the result follows from the non-negativity of $e$, which is a consequence of Lemma \ref{lem:postimvar}.
\end{proof}

\begin{theorem}
	\label{thm:poslo}
	\textbf{(Positive Lower Bound on Solutions)}
	Let $T>0$. Let $y^0 \in L^2(\Omega)$ and $f,g \in L^2(0, T;L^{2}(\Omega))$ be non-negative such that $\|f(t)\|_{\infty}$ and $\|g(t)\|_{\infty}$ are bounded by a constant $C_1>0$ almost everywhere on $t \in [0,T]$. Let $y(\cdot, t)$ be given by 
	\begin{eqnarray}
	y(\cdot,t) &= & \mathcal{T}(t)y^0 - \int_0^t \mathcal{T}(t-s) \Big ( g(\cdot,s)y(\cdot,s) \Big ) ds \nonumber \\
	&~&+ \int_0^t \mathcal{T}(t-s)f(\cdot,s) ds
	\end{eqnarray}
	for all $t \in [0,T]$. Then there exist constants $\tau,\epsilon,\delta>0 $, independent of $y^0$ and $T>0$, such that if $\tau +\delta < T$, then $y(\cdot,t) \geq \epsilon \|y^0\|_1$ for all $t \geq [\tau,\tau+\delta]$.
\end{theorem}
\begin{proof}
	We know from Theorem \ref{thm:origbnd} that there exists a constant $k>0$ and time $T >0$, independent of $y^0$, such that $\mathcal{T}(t)y^0 \geq k \|y^0\|_1 $ for all $t \geq T$. Let $C = -A- \|g\|_{\infty} \mathbf{I}$. Then the semigroup $(\mathcal{T}^{C}(t))$ generated by the operator $C$ is given by $\mathcal{T}(t) = e^{-\|g\|_{\infty}t}\mathcal{T}(t)$ for all $t \geq 0$. The result then follows from Theorem \ref{thm:cmppi}.
\end{proof}

The following theorem states the fundamental result that the PDE \eqref{eq:clpPDEdsicon} conserves mass and maintains positivity.
\begin{theorem}
	\label{thm:solpo}
	Let $
 \kar{\mathbf{y}^0} \in \mathbf{L}^{\infty}(\Omega)$ be non-negative. Then the unique global mild solution of the PDE \eqref{eq:vecpde} is non-negative, and $ \|\mathbf{y}(\cdot,t)\|_1 = \| \mathbf{y}^0\|_1$ for all $t \geq 0$. 
\end{theorem}
\begin{proof}
	The conservation of mass is a simple consequence of taking the inner product of the solution of \eqref{eq:clpPDEdsicon} with a constant function. The positivity property of solutions follows from \cite{duprez2017criterion}[Theorem 1] by noting that, if $\lambda >0$ is large enough, then $G(\y) + \lambda \y \geq  0$ for all $\y \in \eL^2(\Omega)$ that are non-negative. 
\end{proof}

We 
now require some additional notation. For a function $f \in L^2{(\Omega)}$, we define $f_+ := \frac{|f|+f}{2}$, the projection of $f$ onto the set of non-negative functions in $L^2(\Omega)$, and  $f_-:=-\frac{|f|-f}{2}$, the projection of $f$ onto the set of non-positive functions in $L^2(\Omega)$. Given these definitions, we have the following result on partial monotonicity of solutions of the PDE \eqref{eq:clpPDEdsicon}. 

\begin{proposition}
	\label{prop:monprop}
	\textbf{(Partial Monotonicity of Solutions)}
	Let $\kar{\y^0} \in \eL^{\infty}(\Omega)$ be positive. Then, for all $t \geq s \geq 0$, the unique global mild solution of the PDE \eqref{eq:vecpde} satisfies 
	\begin{eqnarray}
	\label{eq:moneq1}
	( y^d - y_2(\cdot,t))_+ ~\leq~  (y^d - y_2(\cdot,s))_+ \\
	( y^d - y_2(\cdot,t))_- ~\geq~  (  y^d - y_2(\cdot,s))_- \label{eq:moneq2} 
	\end{eqnarray}
\end{proposition}
\begin{proof}
	We will only prove the first inequality \eqref{eq:moneq1}. Since $\y^0 \in \eL^{\infty}(\Omega)$, we know that $y_2 \in C([0,1];L^{2}(\Omega))$ and $\|y_2(t) \|_{\infty } $ is uniformly bounded over $[0,T]$. 
	Assume that $y^d-y^0_2$ is non-zero and non-negative on a set $\Omega_1 \subseteq \Omega $ of positive measure. For the sake of contradiction, suppose that there exists $t_2 \in (0,T]$ such that $y_2(\cdot,t_2)$ is greater than $y^d$ on a subset of $\Omega_1$ that has positive Lebesgue measure. Then, due to the fact that $y_2 \in C([0,T];L^2(\Omega))$, there must exist  $ t_1 \in (0,t_2)$ and a measurable set $\Omega_2 \subset \Omega_1$ of positive Lebesgue measure, such that for each $s \in [t_1,t_2]$, $y_2(\mathbf{x},s) \geq y^d(\mathbf{x})$ for almost every $\mathbf{x} \in \Omega_2$, with $y_2(\mathbf{x},t_2) \neq y_2(\mathbf{x},s)$ for almost every $\x \in \Omega_2$ and a subset of $[t_1,t_2]$ with positive Lebesgue measure. However, we know that 
	\begin{eqnarray}
	y_2(\cdot,t) &=& y_2(\cdot,t_1) + \int_{\kar{t_1}}^t F_1(y_2(\cdot,\tau)) y_1(\cdot,\tau)d\tau \nonumber \\
	&~& - \int_{\kar{t_1}}^t F_2(y_2(\cdot,\tau)) \kar{y_2   (\cdot,\tau)} d\tau
	\end{eqnarray}
	for all $ t \in [t_1,t_2]$. 
	This implies that 
	\begin{eqnarray}
	y_2(\x,t) & = & y_2(\x,t_1) + \int_{\kar{t_1}}^t F_1(y_2(\x,\tau)) y_1(\x,\tau)d\tau \nonumber \\
	&~& - \int_{\kar{t_1}}^t F_2(y_2(\x,\tau)) y_2 (\x,\tau)d\tau \nonumber \\
	&=& y_2(\x,t_1)  - \int_{\kar{t_1}}^t r_2(y_2(\x,\tau)-y^d(\x)) \kar{y_2(\x,\tau)} d\tau \nonumber 
	\end{eqnarray}
	for almost every $\x \in \Omega_2$ and for all $t \in [t_1,t_2]$. The function $y_2$  is non-negative due to Theorem \ref{thm:solpo}. Moreover,  $r_2$ is also non-negative    by definition. Thus, we arrive at the contradiction that $y_2(\x,t) \leq y_2(\x,t_1)$ for almost every $\x \in \Omega_1$ and for all $t \in [t_1,t_2]$. Hence, we must have that 
	\begin{equation}
	y_2(\x,t) = y_2(\x,t_1) + \int_{\kar{t_1}}^t  r_1(y_2(\x,\tau)-y^d(\x)) y_1(\x,\tau)d\tau  \nonumber
	\end{equation}
	for almost every $\x \in \Omega_1$ and for all $t \in [0,T]$. This implies that $y_2$ is non-decreasing with time, and is less than or equal to $y^d$ almost everywhere on $\Omega_1$. This proves the first inequality \eqref{eq:moneq1}. 
	Using a similar argument, based on the fact that $r_1$ and $r_2$ are non-negative bounded functions, we can arrive at the second inequality \eqref{eq:moneq2}.
\end{proof}

Using the above proposition, we will establish global asymptotic stability of the system \eqref{eq:clpPDEdsicon} in the $L^1$ norm. Towards this end, we first establish marginal stability of the system about the equilibrium distribution $\y^d$.
\begin{theorem}
	\textbf{($L^1$-Lyapunov Stability)}
	Let $\y^0 \in \eL^{\infty}(\Omega)$ be positive and $\int_{\Omega} \mathbf{y}^0(\x)d\x =1$. For every $ \epsilon>0$, if
	\begin{equation}
	\|\mathbf{y}^0-\mathbf{y}^d\|_1 \leq  \epsilon,
	\end{equation}
	then the solution $\y(\cdot,t)$ of the system \eqref{eq:clpPDEdsicon} satisfies 
	\begin{equation}
	\|\mathbf{y}(\cdot, t)-\mathbf{y}^d\|_1 \leq 2\epsilon
	\end{equation}
	for all $t \geq 0$.
	
\end{theorem}
\begin{proof}
	We know that the solution $\mathbf{y}$  satisfies \[\int_{\Omega}\y(\cdot,t) d\x = \int_{\Omega}y_1(\cdot,t)d\x+  \int_{\Omega}y_2(\cdot,t)d\x =1\]for all $t \in [0,T]$. From Proposition \ref{prop:monprop}, we know that $\|y_2(\cdot,t)-y^d\|_1$ is non-decreasing with time $t$. Hence, $\|y_2(\cdot,t)-y^d\|_1 \leq \epsilon $ for all $t \geq 0$.
	Then, we have that
	\begin{equation}
	\int_{\Omega}y_1(\x,t)d\x +\int_{\Omega}(y_2(\x,t)-y^d(\x))d\x = 1-\int_{\Omega}y^d(\x)d\x \nonumber
	\end{equation}
	for all $t \geq 0$. This implies that
	\begin{eqnarray}
	\int_{\Omega}y_1(\x,t)d\x ~&\leq&~ -\int_{\Omega}(y_2(\x,t)-y^d(\x))d\x \nonumber \\
~&\leq&~\|y_2(\cdot,t) -y^d \|_1 
	~\leq~ \epsilon \nonumber
	\end{eqnarray}
	for all $t \geq 0$. This concludes the proof.
\end{proof}

\begin{proposition}
	\label{prop:posytwo}
	Let $\y^0 \in \eL^{\infty}(\Omega)$ be non-negative and $\|\mathbf{y}^0\|_1 =1$. Then the solution $\y$ of the PDE \eqref{eq:clpPDEdsicon} satisfies $\lim_{t\rightarrow \infty}\|(\kar{y_2(\cdot,t)}-y^d)_{+}\|_{\infty} = 0$.
\end{proposition}
\begin{proof}
	Suppose that, for the sake of contradiction, this is not true. Then, due to the partial monotonicity property of the solution $\y$ stated in Proposition \ref{prop:monprop}, there exists a subset $\Omega_1 \subseteq \Omega$ of positive measure, and a parameter $\epsilon>0$, such that $\kar{y_2(\mathbf{x},t)}-y^d(\x) \geq \epsilon$ for almost every $\kar{\mathbf{x}} \in \Omega_1$ and all $t \geq 0$. However, we know that  
	\begin{eqnarray}
	y_2(\x,t) &=& y_2(\x,\kar{0})  - \int_{\kar{0}}^t \kar{F_2}(y_2(\x,\tau)) y_2(\x,\tau)d\tau \nonumber \\
	&=& y_2(\x,\kar{0})  - \int_{\kar{0}}^t r_2(y_2(\x,\tau)-y^d(\x)) y_2(\x,\tau)d\tau \nonumber
	\end{eqnarray}
	for almost every $\x \in \Omega_1$ and for all $t \geq 0$. We know that the function $r_2$ is non-zero and continuous on the open interval $t \in (0,\infty)$. Hence, there must exist $\delta >0$ such that 
	\begin{eqnarray}
	y_2(\x,t)  ~ & \leq&~  y_2(\x,0)  - \int_0^t \delta y_2(\x,\tau)d\tau \nonumber \\
	 ~&\leq&~ y_2(\x,0) 
	- \delta  \int_0^t(y^d(\x)+\epsilon)d\tau
	\end{eqnarray}
	for almost every $\x \in \Omega_1$ and for all $t \geq 0$. This leads to a contradiction.
\end{proof}


Finally, we can establish attractivity of the equilibrium point $\y^d \in L^{\infty}(\Omega)$. Towards this end, we first prove in the lemma below 
that the density of the agents in the state $Y(t)=0$ (i.e. the state of motion, given by $y_1$, 
 must converge to $0$ eventually.

\begin{lemma}
	Let $\y^0 \in \eL^{\infty}(\Omega)$ be non-negative and $\|\mathbf{y}^0\|_1 =1$. Then $\lim_{t \rightarrow \infty}\|y_1(\cdot,t)\|_1 = 0.$ Hence, $\lim_{t \rightarrow \infty}\|y_2(\cdot,t)\|_1 = 1.$
	\label{lem:linfconvpos}
\end{lemma}
\begin{proof}
	Suppose that, for the sake of contradiction, this is not true. Then there exists $\epsilon_1 >0$ and a sequence of increasing time instants $(t_i)_{i=1}^{\infty}$ such that $\lim_{i \rightarrow \infty}t_i = \infty$ and $\|\kar{y_1(\cdot,t_i)}\|_1 \geq \epsilon_1$ for all $i \in \mathbb{Z}_{+}$. From Theorem \ref{thm:poslo}, we know that this implies that there exist constants $\tau,\epsilon_2,\delta>0 $ such that \kar{$y_1(\cdot,t) \geq \epsilon_2 \|y_1(\cdot,t_i)\|_1  \geq \epsilon_1\epsilon_2 $} for all $t \kar{~\in~} [t_i,t_i+\delta]$, for all $i \in \mathbb{Z}_+$. Without loss of generality, we can assume that $t_{i+1}-t_{i} > \delta$ for all $i \in \mathbb{Z}_+$. Let $\Omega_1 \subseteq \Omega$ be the subset of largest measure such that $y^0_2(\mathbf{x}) \geq y^d(\mathbf{x})$ for all $\x \in \Omega_1$. Then, from the partial monotonicity property of the solution $\y$ (Proposition \ref{prop:monprop}), we have that, for each $i \in \mathbb{Z}_{+}$, 
	\begin{eqnarray}
	y_2(\x,t_{i} +\delta) \hspace{-2mm} &=& \hspace{-2mm} y_2(\x,0)  + \int_0^{t_{i} + \delta} F_1(y_2(\x,\tau)) y_1(\x,\tau)d\tau \nonumber \\
	&\geq & \hspace{-2mm} y_2(\x,0) \nonumber \\
	&~& \hspace{-2mm} + \sum_{j = 1}^{i}\int_{t_i}^{t_{i}+\delta} r_1(y_2(\x,\tau)-y^d(\x)) y_1(\x,\tau)d\tau  \nonumber 
	\end{eqnarray}
	for almost every $\x \in \Omega_1$. 
	This implies that $\lim_{i \rightarrow \infty}\|(\kar{y_2(\cdot,t_i)}-y^d)_{-}\|_{\infty} = 0$. However, we know that $\|\y(\cdot, t)\|_1 =1$ for all $t \geq 0$. This, along with the fact that $\lim_{t \rightarrow \infty}\|(\kar{y_2(\cdot,t)}-y^d)_{+}\|_{\infty} = 0$ (\kar{Proposition \ref{prop:posytwo}}) and the assumption that $\|\kar{y_1(\cdot,t_i)}\|_1 \geq \epsilon_1$ for all $i \in \mathbb{Z}_{+}$, leads to  a contradiction.  
\end{proof}

Using the partial monotonicity property of solutions established in Proposition \ref{prop:monprop} and the result in Lemma \ref{lem:linfconvpos}, we now obtain the following global asymptotic stability result. 

\begin{theorem}
	\textbf{($L^1$-Global Asymptotic Stability)}
	Let $\y^0 \in L^{\infty}(\Omega)$ be non-negative and $\|\mathbf{y}^0\|_1 =1$. Then $\lim_{t \rightarrow \infty}\|\mathbf{y}(\cdot,t)-\mathbf{y}^d\|_1=0$, and hence the system \eqref{eq:clpPDEdsicon} is globally asymptotically stable about the target equilibrium distribution $\mathbf{y}^d$.
\end{theorem}
\begin{proof}
	Let $\Omega_1 = \lbrace \x \in \Omega ; ~y^0_2(\x) \geq y^d(\x)  \rbrace$. Let $\Omega_2 = \Omega - \Omega_1 $. From \kar{Proposition \ref{prop:posytwo}}, we know that $\lim_{t \rightarrow \infty}\|y(\cdot,t)|_{\Omega_1}-y^d|_{\Omega_1}\|_\infty =0$, where $\cdot|_{\Omega}$ denotes the restriction operation. This implies that $ \lim_{t \rightarrow \infty}\|y(\cdot,t)|_{\Omega_1}-y^d|_{\Omega_1}\|_1 =0$. We also know from \kar{Proposition \ref{prop:posytwo}} that 
	\begin{equation}
	\lim_{t \rightarrow \infty} \int_{\Omega_1} \big ( y_2(\x,t)-y^d(\x)\big )d\x+ \int_{\Omega_2} \big( y_2(\x,t)-y^d(\x) \big)d\x  = 0 
	\end{equation}
	This implies that 
	\begin{equation}
\lim_{t \rightarrow \infty} \int_{\Omega_2}\big (y_2(\x,t)-y^d(\x)\big )d\x =0 	\end{equation}
	From Proposition \ref{prop:monprop}, we know that $y_2(\cdot,t) \leq y^d$ almost everywhere on $\Omega_2$ and for all $t \geq 0$. Hence, we 
 conclude that 
	\begin{eqnarray}
	\lim_{t \rightarrow \infty} \|y_2(\cdot,t)-y^d\|_1 = 	\lim_{t \rightarrow \infty} \int\limits_{\Omega_1}|y_2(\x,t)-y^d(\x)|d\x \nonumber  \\ 
	  + \int\limits_{\Omega_2} |y_2(\x,t)-y^d(\x)|d\x    \nonumber
	\end{eqnarray}
	From this, along with the fact that $ \lim_{t \rightarrow \infty}\|y_1(\cdot,t)\|_1=0$, we arrive at our result.
\end{proof}

\section{SIMULATIONS}
\label{sec:numsim}

In this section, we validate the control laws presented in Sections \ref{sec:subell} and \ref{sec:Disco} with numerical simulations. The SDEs \eqref{eq:SDEhol} and \eqref{eq:litSDE2} were simulated using the method of Wong-Zakai approximations  \cite{twardowska1996wong}. The diffusion and reaction parameter values used in each simulation were chosen with the goal of shortening the duration of the simulation on a case-by-case basis. Hence, different parameter values were chosen for each of the examples below. In practice, these parameters would need to be chosen according to the physical constraints on the system and the objectives of the user, such as optimizing the rate of convergence to the target density or controlling the variance of the agent density around the target density.

\subsection{Density Control without Agent Interactions}

In this subsection, we simulate the control approach presented in Section \ref{sec:subell}. In the following example, we simulate  the SDE \eqref{eq:SDEhol} with the control laws $v_i =0$ and $u_i(\mathbf{x})= D/y_d(\mathbf{x})$, where $D$ is a diffusion coefficient. The generator of the process is  given by the operator in \eqref{eq:diffgena}.
\begin{example}
\label{eg:egbini}
\textbf{Brockett integrator}
\end{example}
In this example, we consider the case where each agent's motion evolves according to the Brockett integrator, which has been well-studied in the control theory literature  \cite{brockett1983asymptotic,astolfi1998discontinuous}. The control vector fields for this system are the following:
\begin{eqnarray}
\label{eq:brock}
X_1(\mathbf{x}) = \frac{\partial}{\partial x_1} - x_2\frac{\partial}{\partial x_3},  ~~~~
X_2(\mathbf{x}) = \frac{\partial}{\partial x_2} + x_1\frac{\partial}{\partial x_3} 
\end{eqnarray}
 The Lie bracket of the two vector fields is given by 
\begin{eqnarray}
[X_1,X_2](\mathbf{x}) =X_1X_2 -X_2X_1 = 2 \frac{\partial }{\partial x_3}
\end{eqnarray}
for all $\mathbf{x} \in \kar{\mathbb{R}^n}$. 
Hence, we have that
${\rm span~}\lbrace X_1(\mathbf{x}),X_2(\mathbf{x}),[X_1,X_2](\mathbf{x}) \rbrace = T_\mathbf{x}\mathbb{R}^3$
and therefore, the system is bracket generating. We define the domain $\Omega = [0,100]^3$.
 The target distribution 
 is given by 
$y^d = c~ [ \sum_{i=1}^8\mathbf{1}_{B_{\mathbf{x}_i}} +0.001],$
where $c > 0$ is a normalization constant that makes $y^d$ a probability density, $\mathbf{1}_S$ denotes the indicator function of a set $S$, and $B_{\mathbf{x}_i}$ denotes a ball of radius $12.5$ centered at $\mathbf{x}_i$, $i=1,...,8$, defined as  $\lbrace \mathbf{x}_1,...,\mathbf{x}_8\rbrace = \lbrace [ 25~25~25]^T, [25~25~75]^T, [25~75~75]^T,...,[75~75~75]^T\rbrace. $

The  positions  of $N_p=   10,000$ agents  are generated from a stochastic simulation of the SDE \eqref{eq:SDEhol} and  plotted at three times $t$ in Fig. \ref{fig:BI1}.  The figure shows that at  time $t= 100~s$,  the  distribution  of  the swarm over the sphere is close to the target density. \kar{In Fig. \ref{fig:L1nonin}, it can be seen that the $L_1$ norm of the difference between the current distribution and the target distribution decreases over time.}

\begin{figure*}
	\centering
	\begin{subfigure}[t]{0.3\textwidth}
		\centering
		\includegraphics[width=\textwidth]{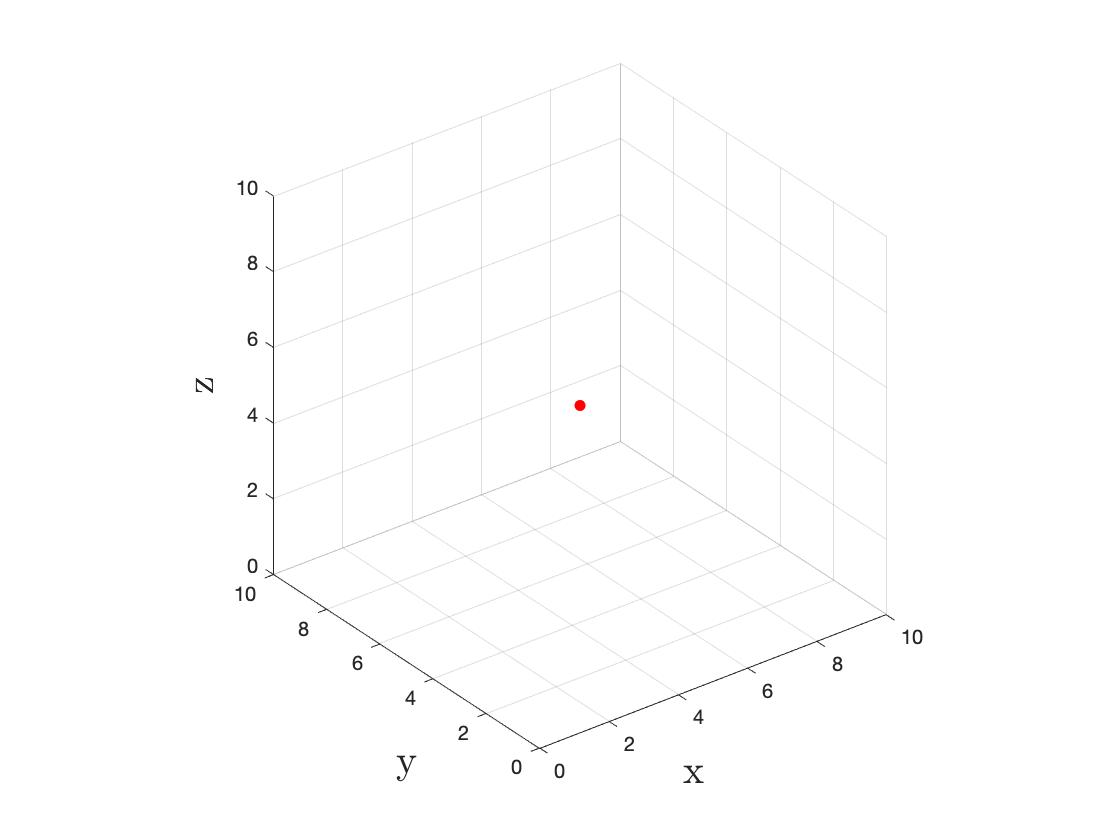}	
		\caption{$t=0$ s}
	\label{subfig:hni1}
	\end{subfigure}%
	~
	\begin{subfigure}[t]{0.3\textwidth}
		\centering
		\includegraphics[width=\textwidth]{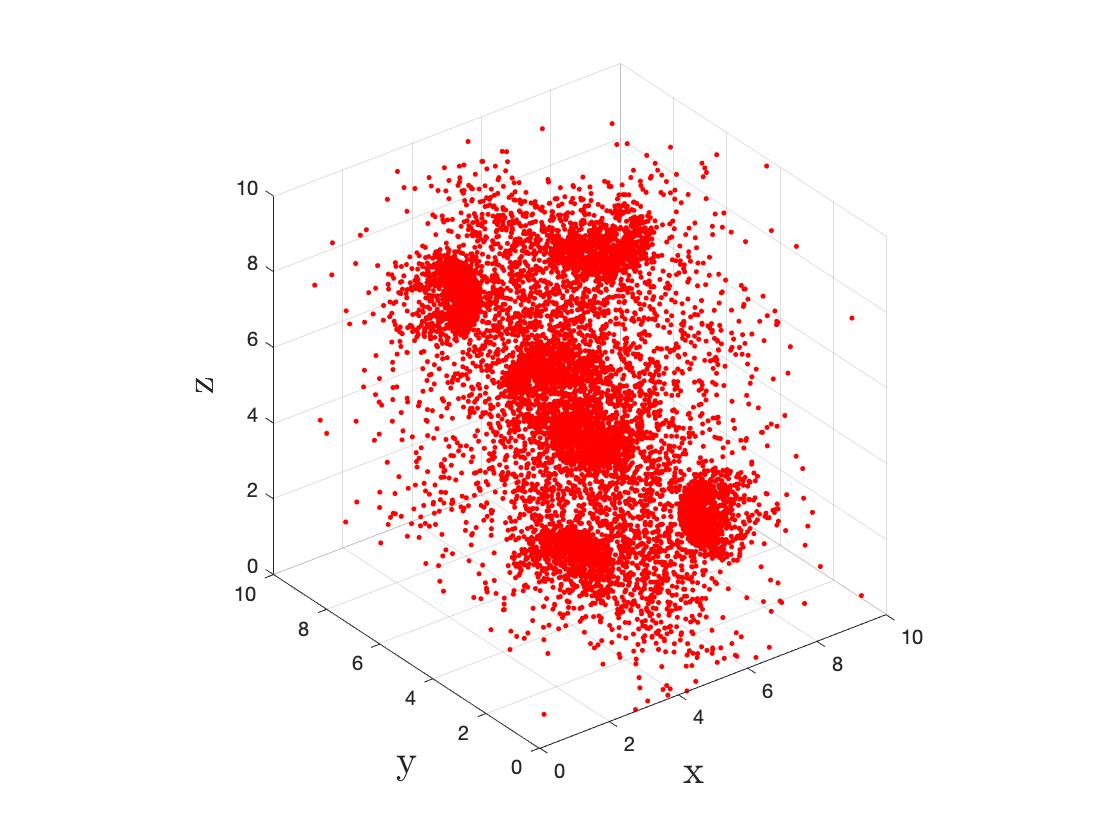}
		\caption{$t=10$ s}
		\label{subfig:hni2}
	\end{subfigure}
	~
	\begin{subfigure}[t]{0.3\textwidth}
		\centering
		\includegraphics[width=\textwidth]{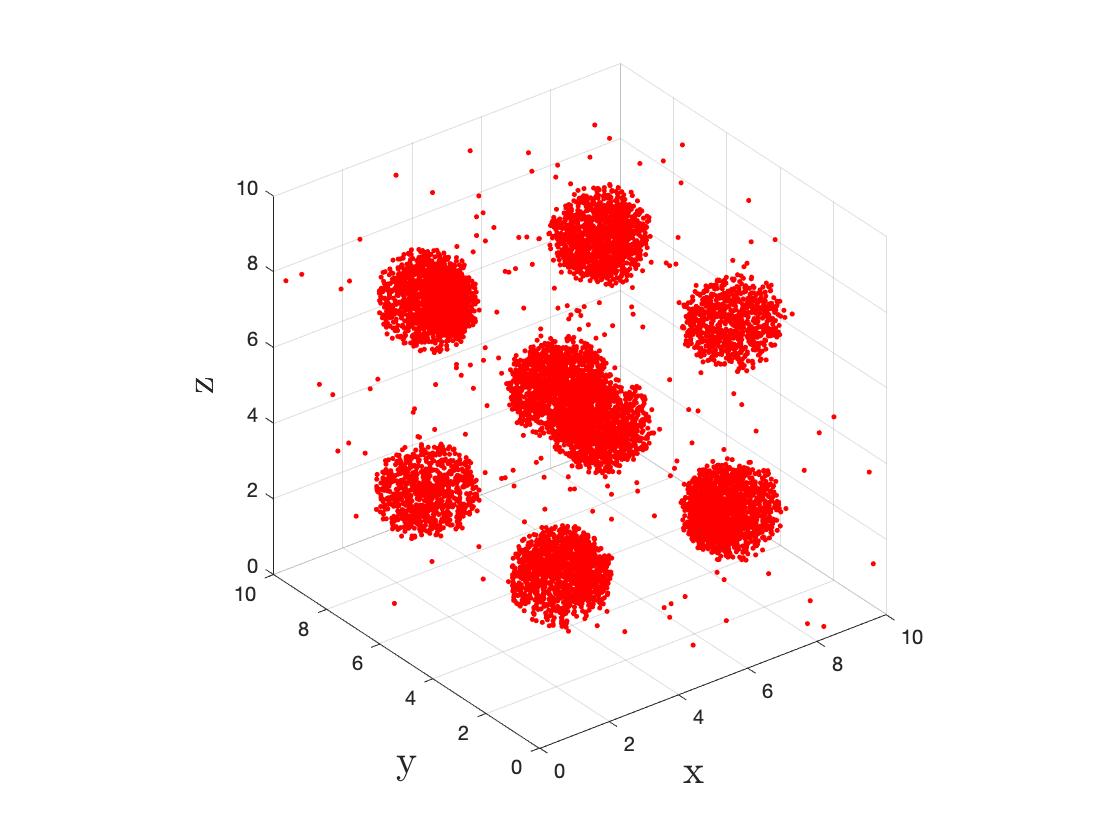}
		\caption{$t=100$ s}
			\label{subfig:hni3}
	\end{subfigure}
	\caption{\textbf{(Brockett integrator without agent interactions)} Stochastic coverage by $N_p= 10,000$ agents (in red) at three 
times $t$, following the linear diffusion model \eqref{eq:Mainsysan}.} 
\label{fig:BI1}
\end{figure*}

\begin{figure}
	\centering
		\includegraphics[scale=0.2]{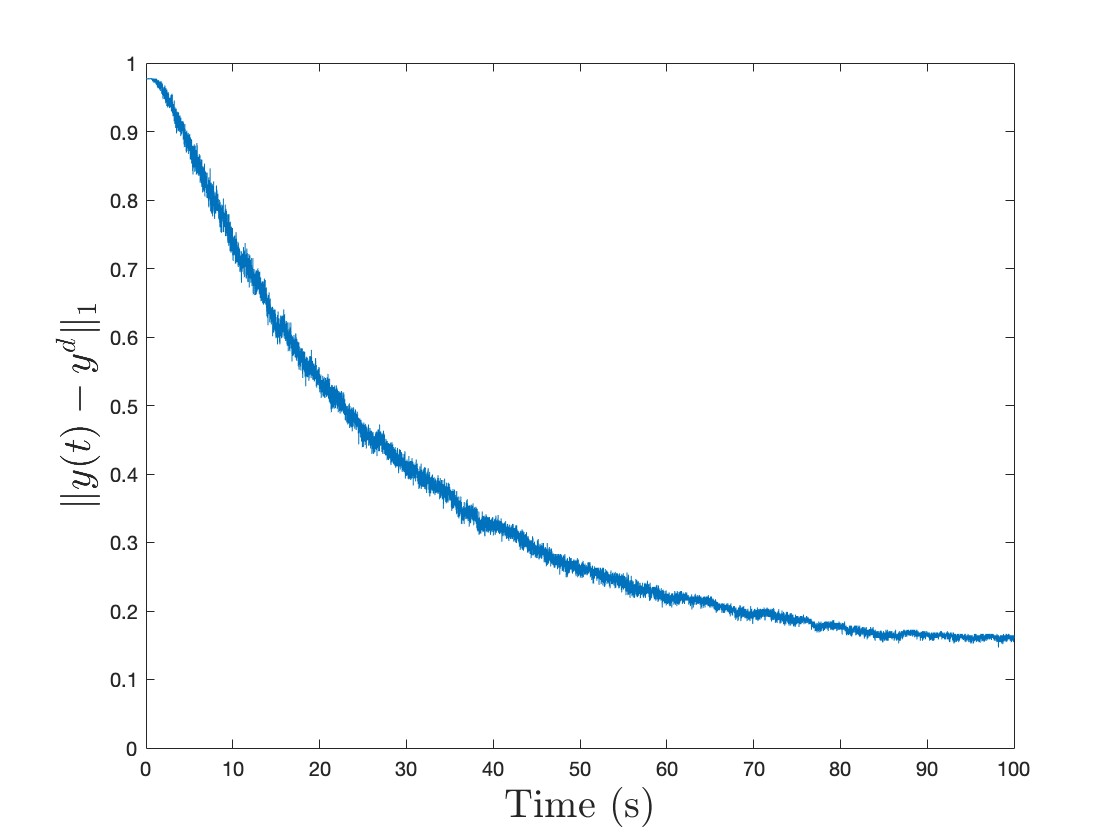}	
\caption{\kar{\textbf{(Brockett integrator without agent interactions)} Time evolution of the $L_1$ norm of the difference between the target distribution and the agent distribution that evolves according to the linear diffusion model \eqref{eq:Mainsysan}.}}
\label{fig:L1nonin}
\end{figure}

\begin{example}
\textbf{Nonholonomic system on SO(3)} 
\end{example}
In this example, the agents evolve according to a nonholonomic system on the special orthogonal group of rotation matrices, denoted by $SO(3) = \{ \mathbf{B} \in \mathbb{R}^{3 \times 3} ; \mathbf{B}^T\mathbf{B} = \mathbf{I},~\det{\mathbf{B} = 1} \}\}$. The tangent space of $SO(3)$ at the origin, the Lie algebra corresponding to this Lie group, is denoted by $so(3)$. A basis of $so(3)$ is given by the following matrices:
\begin{eqnarray}
\label{eq:genma}
\mathbf{B}_1 = 
\begin{bmatrix}
0 & -1 & 0 \\
1 &  0  & 0 \\
0 & 0  & 0 
\end{bmatrix},
~
\mathbf{B}_2 =
\begin{bmatrix}
0 & 0& 1 \\
0 &  0  & 0 \\
-1 & 0  & 0 
\end{bmatrix},
\nonumber \\
\mathbf{B}_3 =
\begin{bmatrix}
0 & 0& 0 \\
0 &  0  & -1 \\
0 & 1  & 0 
\end{bmatrix}.
\end{eqnarray}
Each of these matrices defines a left-invariant vector field $X_i$ given by 
	\begin{equation}
(X_i f)(\mathbf{x}) = \lim_{t \rightarrow 0}\frac{f(e^{t\mathbf{B}_i}\x) - f(\x)}{t}
\end{equation}
for all $\mathbf{x} \in SO(3)$ and all $f \in C^{\infty}(SO(3))$. We  assume that each agent can  control its motion along the vector fields $\{ X_1,X_2 \}$. It can be verified that 
\begin{align}
&{\rm span~}\lbrace X_1(\mathbf{x}),X_2(\mathbf{x}),[X_1,X_2](\mathbf{x}) \rbrace \\ \nonumber
&={\rm span~}\lbrace X_1(\mathbf{x}),X_2(\mathbf{x}),-X_3(\mathbf{x})\rbrace= T_\mathbf{x}SO(3)
\end{align}
The target density (with respect to the Haar measure) is given by
\begin{equation}
y^d(\mathbf{A})=c( A^2_{11} + A^2_{22}+ A^2_{33})
\end{equation} 
for all $\mathbf{A} \in SO(3)$, where $c>0$ is the normalization constant such that $y^d$ is a probability density.

Since $SO(3)$ is a $3$-dimensional manifold and cannot be embedded in $\mathbb{R}^3$, visualizing the agent positions and the target probability density is not straightforward. We adopt the approach suggested in \cite{lee2018bayesian} for such a visualization. Any element of $\mathbf{O}\in SO(3)$ acts on an element $\mathbf{e}$ of the $2$-dimensional sphere $S^2 = \{\mathbf{x} \in \mathbb{R}^3; ~\mathbf{x}^T\mathbf{x} = 1\}$ through matrix-vector multiplication, resulting in an element $\mathbf{Oe}\in S^2$. If $\mathbf{Z}(t)$ is the solution to the process \eqref{eq:SDEhol}, then $\mathbf{Z}(t)\mathbf{e}_i$  is a process on the sphere. Let $\mathbf{e}_1= [1~~0~~0]^T$, $\mathbf{e}_2= [0~~1~~0]^T$, and $\mathbf{e}_3= [0~~1~~0]^T$. Note that $\mathbf{Oe}_i$ is the $i^{th}$ column of $\mathbf{O}$. If the density of $\mathbf{Z}(t)$ converges to $y^d$, then the density of $\mathbf{Z}(t)\mathbf{e}_i$ must converge to the target density on the sphere $S^2$ given by $y^d_i(\mathbf{x}) = \int_{\mathbf{Oe}_i = \mathbf{x}}y^d(\mathbf{O})d\mathbf{O} $ for all $\mathbf{x} \in S^2$.
The action of the matrices $\mathbf{Z}_j(t)$ on the vectors $-\mathbf{e}_2$ and $\mathbf{e}_3$ is shown in Fig. \ref{fig:SO31} and Fig. \ref{fig:SO32}, respectively.  The target densities $y_d^1$ and $y_d^2$, shown in  Fig. \ref{fig:SO31} and Fig. \ref{fig:SO32},  are  depicted  on  the  surface  of  the  sphere  using  a color  density  plot.  Blue  regions  are  assigned  a  low  target density of agents, while yellow regions are assigned a high target density. The agent positions are superimposed on   the  density  plot  to  enable  comparison between  the  actual  and  target  densities.  As Fig. \ref{fig:SO31} and Fig. \ref{fig:SO32}  show, at  time $t= 100~s$,  the  distribution  of  the swarm over the sphere is close to the target density.

\begin{figure*}%
	\centering
	\subfloat[$t =0 ~s$]{%
		\includegraphics[width=0.32\textwidth]{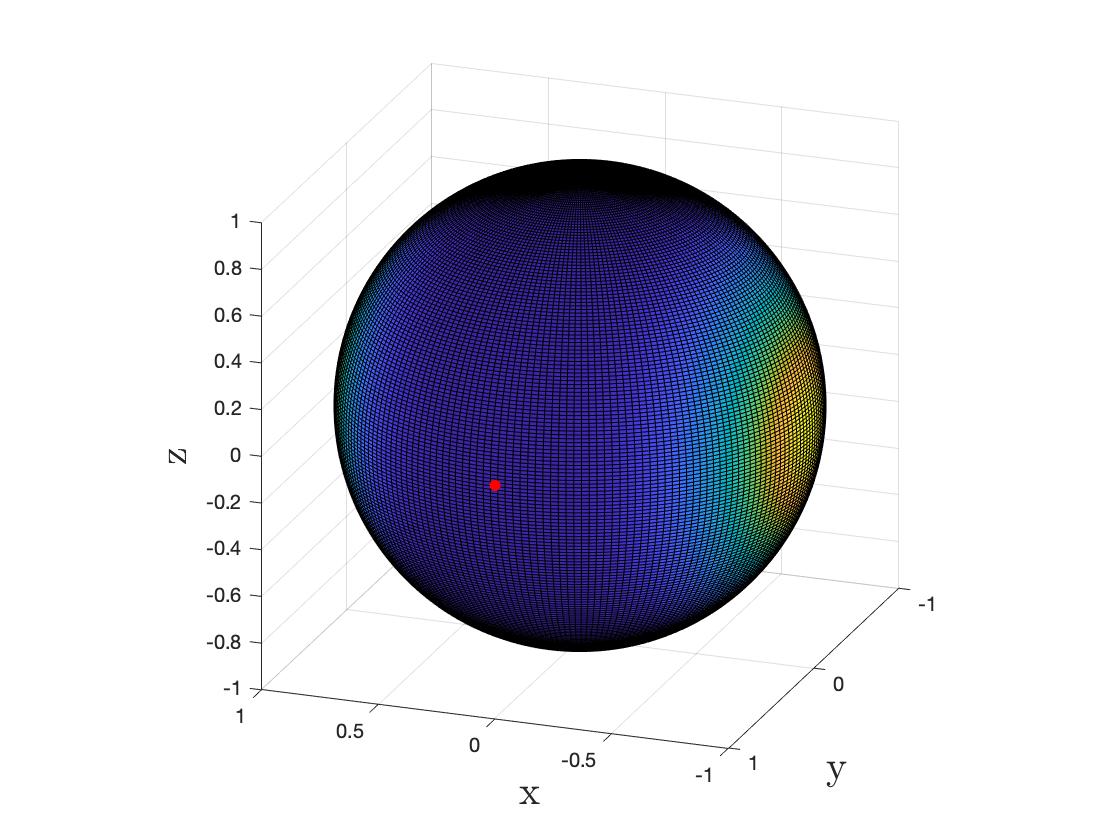}
		\label{subfig:SO311}%
	}%
	\hspace{-0.5em}%
	\subfloat[$t =10~s$]{%
		\includegraphics[width=0.32\textwidth]{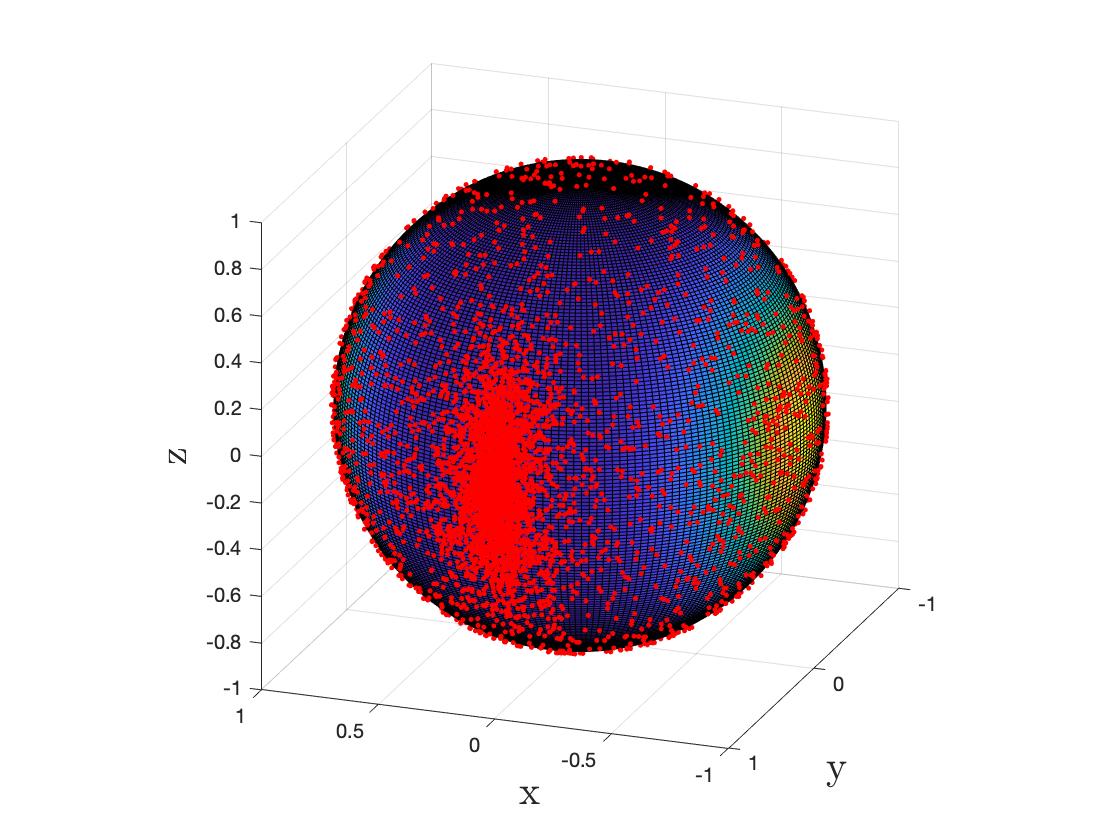}
		\label{subfig:SO312}%
	}%
	\centering
\subfloat[$t =100 ~s $]{%
	\includegraphics[width=0.32\textwidth]{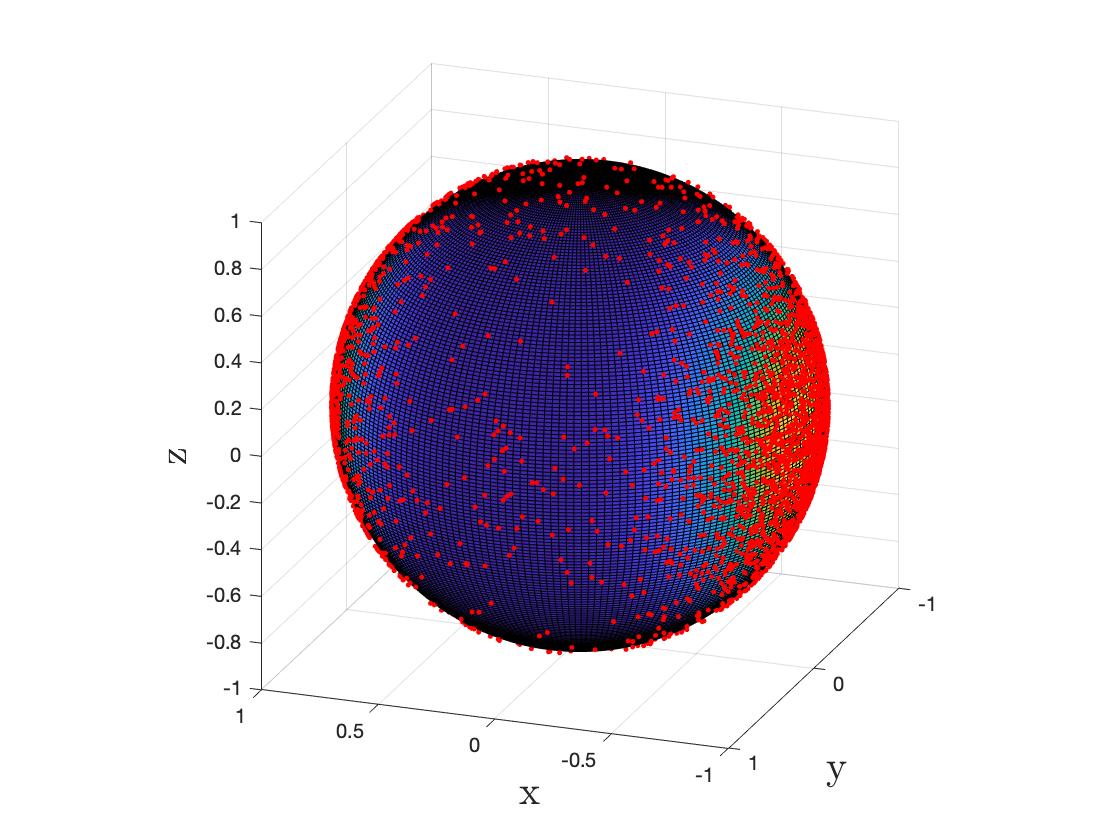}
	\label{subfig:SO313}%
}%
	\caption{\textbf{(Nonholonomic system on SO(3))} Stochastic coverage of $SO(3)$ by $N= 5,000$ agents (in red) at different times $t$, following the linear diffusion model \eqref{eq:Mainsysan}. This plot shows the time evolution of the action of the agents' matrices $\mathbf{Z}_j(t)$ on $-\mathbf{e}_2 = [0 ~~-1~~0]^T$ on the sphere $S^2$.}
			\label{fig:SO31}
\end{figure*}
\begin{figure*}%
	\centering
	\subfloat[$t =0 ~s$]{%
		\includegraphics[width=0.32\textwidth]{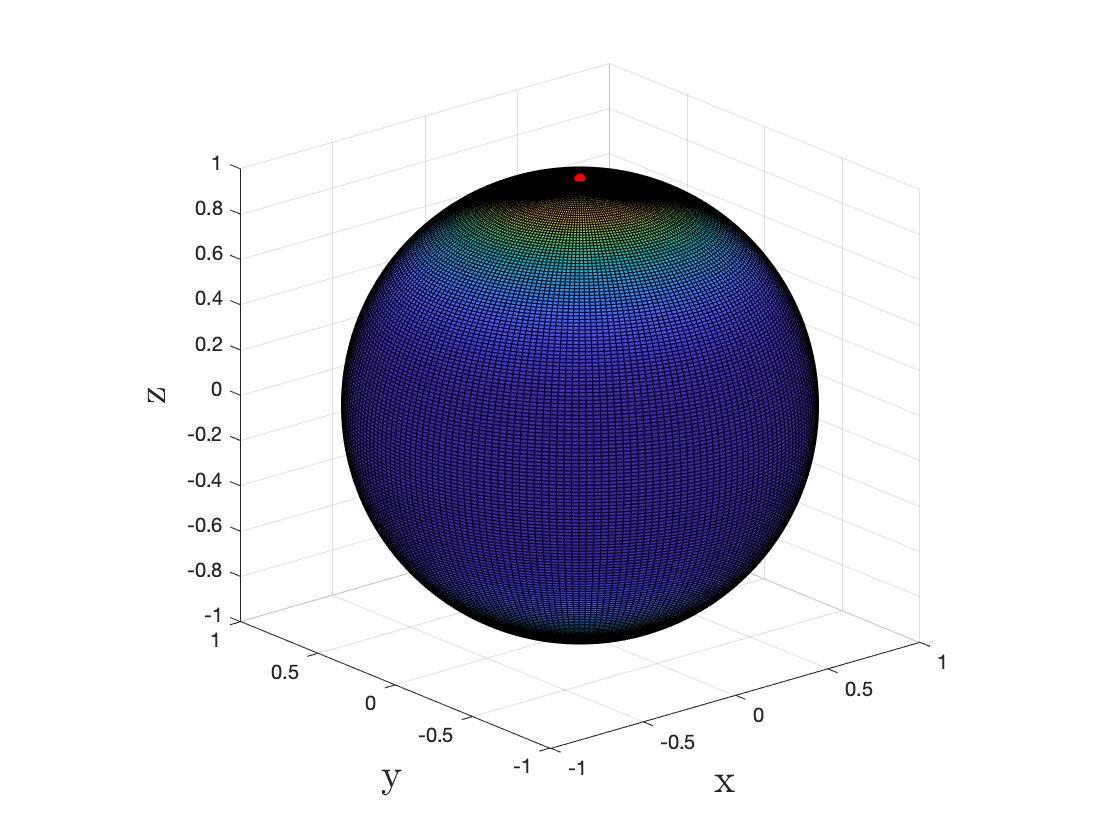}
		\label{subfig:SO321}%
	}%
	\hspace{-0.5em}%
	\subfloat[$t =10~s$]{%
		\includegraphics[width=0.32\textwidth]{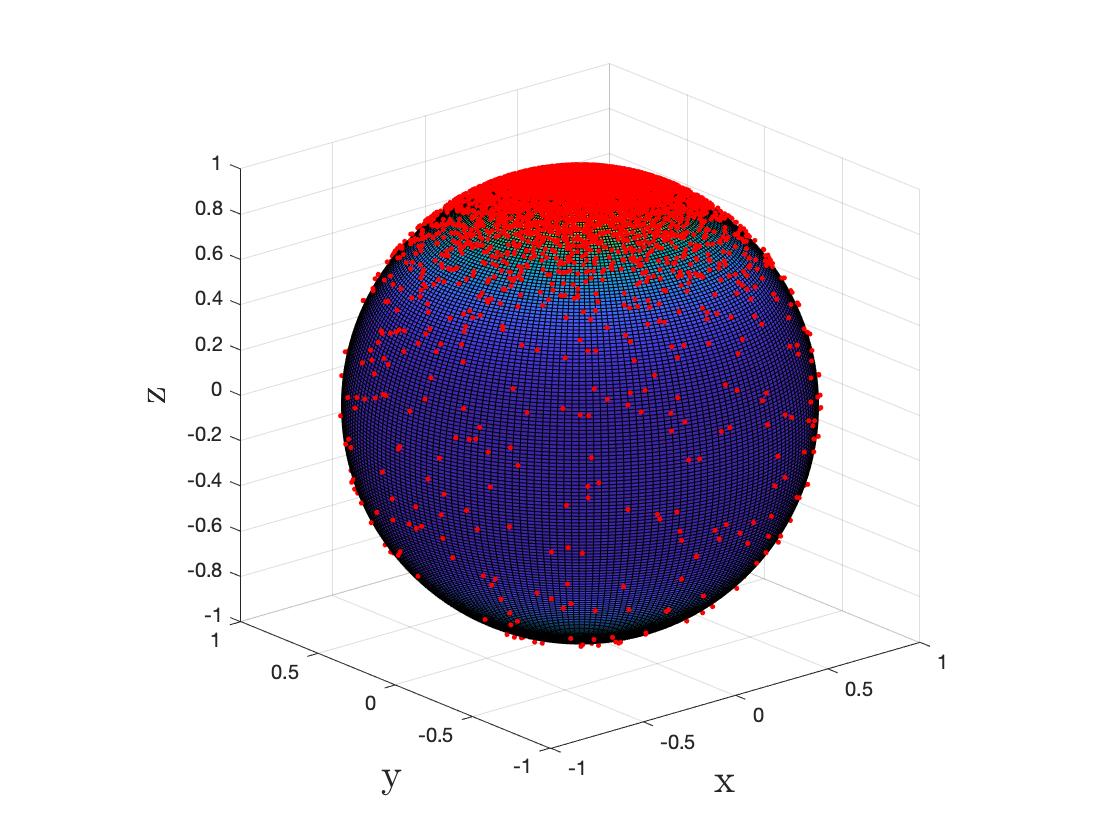}
		\label{subfig:SO322}%
	}%
	\centering
	\subfloat[$t =100 ~s $]{%
		\includegraphics[width=0.32\textwidth]{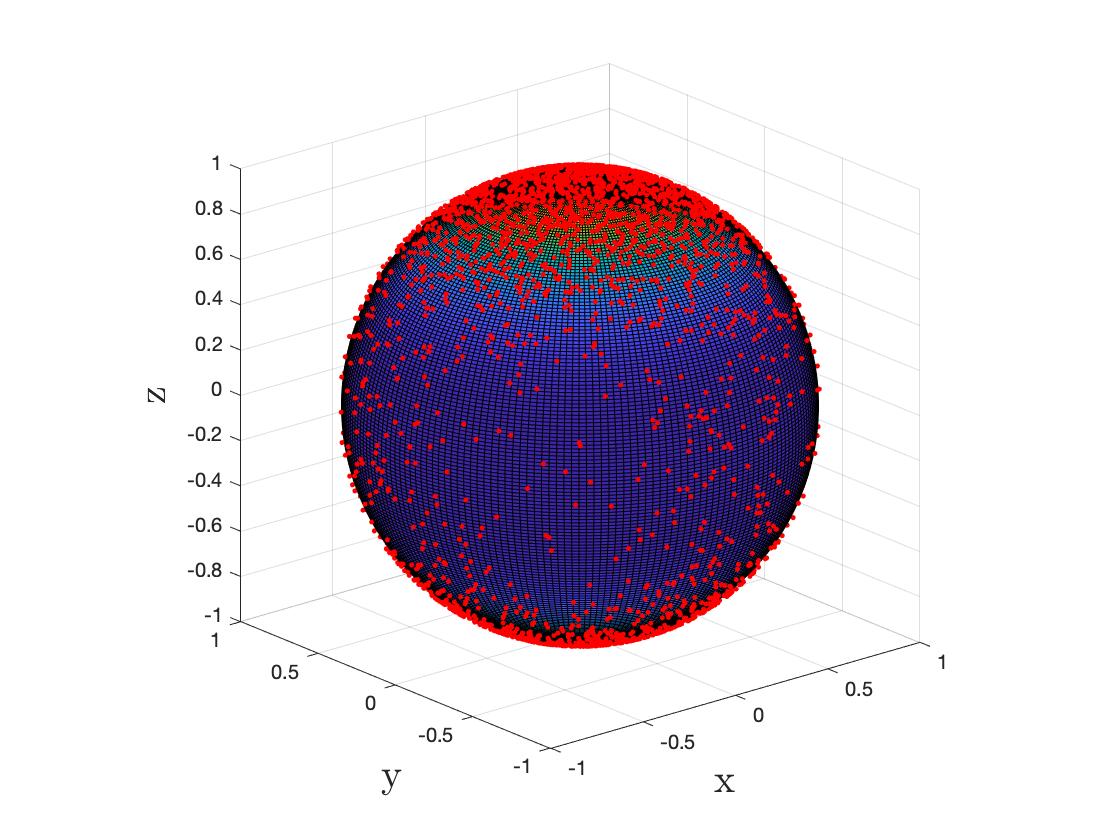}
		\label{subfig:SO323}%
	}%
	\caption{\textbf{(Nonholonomic system on SO(3))} Stochastic coverage of $SO(3)$ by $N= 5,000$ agents (in red) at different times $t$, following the linear diffusion model \eqref{eq:Mainsysan}. This plot shows the time evolution of the action of the agents' matrices $\mathbf{Z}_j(t)$ on $\mathbf{e}_3 = [0 ~~0~~1]^T$ on the sphere $S^2$.}
		\label{fig:SO32}
\end{figure*}


\subsection{Density Control with Agent Interactions}

In this subsection, we simulate the control approach 
in Section \ref{sec:Disco}. The functions $r_i$ in Eq. \eqref{eq:reacp} are chosen to be
\kar{
\begin{equation}
	\label{eq:reac}
r_1(\mathbf{x} )= \begin{cases} -k \mathbf{x}   ~~ {\rm if}~\mathbf{x} < 0  \\
0 ~~~~ {\rm if}~\mathbf{x} =0\end{cases} \end{equation}
\begin{equation}
	\label{eq:reac2}
r_2(\mathbf{x} )= \begin{cases} k \mathbf{x}   ~~ {\rm if}~\mathbf{x} >0  \\
0 ~~~~ {\rm if}~\mathbf{x} =0\end{cases}   \end{equation}
}
for all $\mathbf{x} \in \Omega$,
where $k$ is a positive scaling constant.
	Since we simulate a finite number of agents, instead of the density $y_1$, the agents use the empirical measure $\frac{1}{N_p}\sum_{i=1}^{N_p}\delta_{\mathbf{x}_i(t)}$ to compute their transition rates. However, the empirical measure is not absolutely continuous with respect to the Riemannian volume and does not have a density. Therefore, the agents use the regularized approximation of the measure $\frac{1}{N_p}\sum_{i=1}^{N_p}\delta_{\mathbf{x}_i(t)}$,  given by
	\begin{equation}
	\label{eq:kern}
	\tilde{\rho} (\mathbf{x},t) = c(\epsilon)\frac{1}{N_p}\sum_{i=1}^{N_p}K_{\epsilon}(\mathbf{x},\mathbf{x}_i(t)) 
	\end{equation}
	for all $\mathbf{x}\in M$, where the kernel function $K_{\epsilon}$ is chosen such that $\lim_{\epsilon \rightarrow 0}c(\epsilon)K_{\epsilon}(\cdot,\mathbf{y}) =\delta_{\mathbf{y}}$ for each $\mathbf{y} \in \Omega$, and the function $c(\epsilon)$ is a normalization parameter defined such that $c(\epsilon)\int_{M}K_{\epsilon}(\mathbf{x},\mathbf{y})d\mathbf{x} = 1$. For each of the examples below, we will specify the kernel function used. The positions of each agent are generated according to the SDE \eqref{eq:litSDE2}. The transition rates $q_i(\mathbf{x},t)$ of the agents are defined as
	\begin{equation}
	    q_i (\mathbf{x},t) = r_i \Big(  c(\epsilon)\frac{1}{N_p}\sum_{i=1}^{N_p}K_{\epsilon}(\mathbf{x},\mathbf{x}_i(t)) -y^d(\mathbf{x}) \Big ).
	\end{equation}

\begin{example}
	\textbf{Brockett integrator}
	\label{eg:egbini2}
\end{example}

 In this example, each agent moves according to Eq. \eqref{eq:litSDE2} with the control vector fields as defined in Example \ref{eg:egbini}.
 The kernel function is given by: 
	\begin{equation}
	\label{eq:kern1}
K_{\epsilon}(\mathbf{x},\mathbf{y})  =\begin{cases} 
	\exp{\frac{-1}{1-(|\mathbf{x}-\mathbf{y}| /\epsilon)^2}}~{\rm ~if~} |\mathbf{x}-\mathbf{y}|< \epsilon\\
                     0 ~~{\rm otherwise}
                    \end{cases}
\end{equation}
for all $\x ,\mathbf{y} \in \mathbb{R}^3$. The target density is set to $y^d =  c \sum_{i=1}^8\mathbf{1}_{B_{\x_i}}$, 
similar to the $y^d$ defined  
in Example \ref{eg:egbini}, where $c > 0$ is a normalization constant. 
Note that in this example, the probability density is allowed to take value equal to $0$ in certain regions of the the domain, unlike in Example \ref{eg:egbini}. The reaction constant in Eqs. \eqref{eq:reac}-\eqref{eq:reac2} is set to $k = 500$, and the parameter of the kernel in \eqref{eq:kern1} is defined as $\epsilon = 5$. 


\begin{figure*}
	\centering
	\begin{subfigure}[t]{0.3\textwidth}
		\centering
		\includegraphics[width= \textwidth]{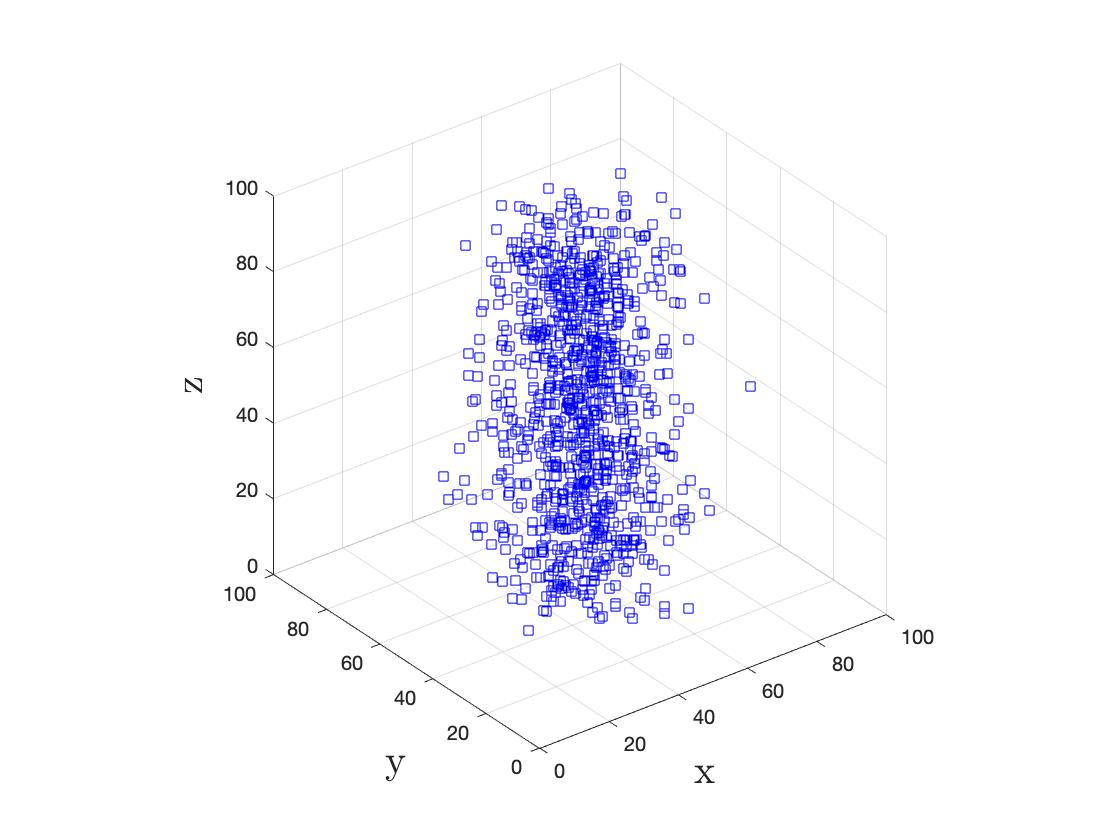}
		\caption{$t=0$ s}
		\label{subfig:Heis1}
	\end{subfigure}%
	~
	\begin{subfigure}[t]{0.3\textwidth}
		\centering
		\includegraphics[width= \textwidth]{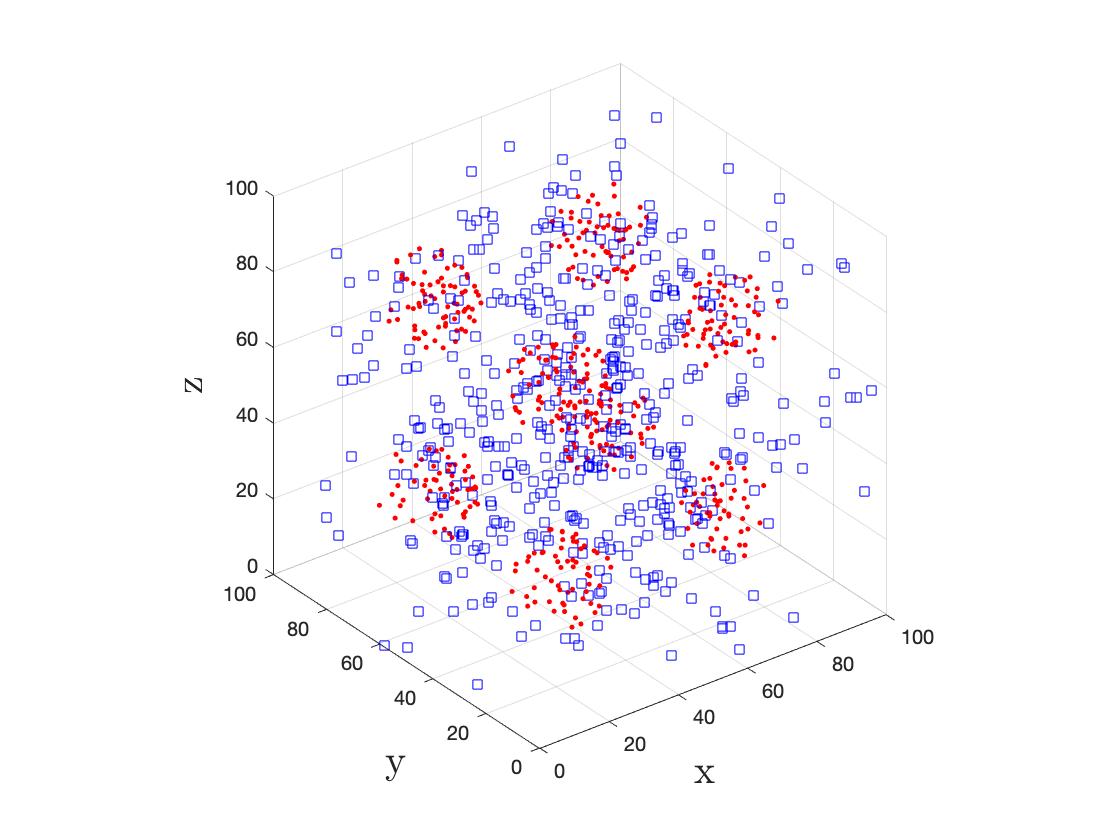}
		\caption{$t=10$ s}
		\label{subfig:Heis2}
	\end{subfigure}
	~
	\begin{subfigure}[t]{0.3\textwidth}
		\centering
		\includegraphics[width= \textwidth]{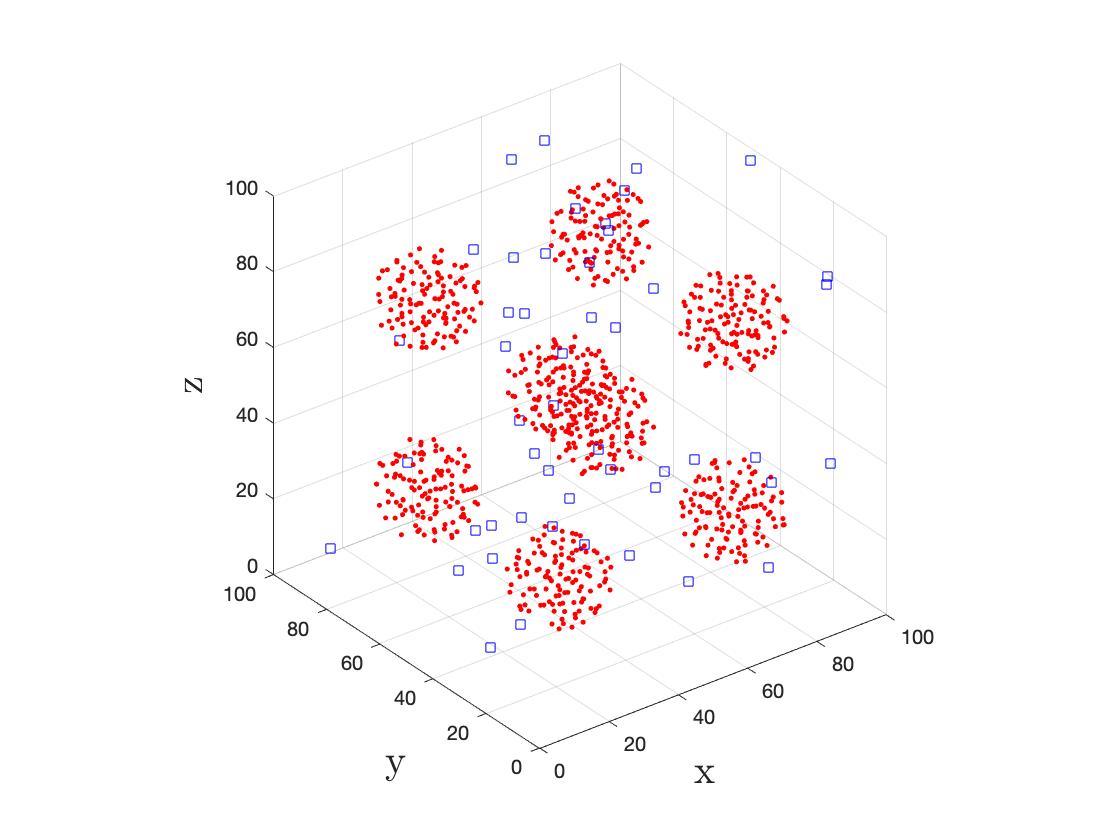}
		\caption{$t=100$ s}
			\label{subfig:Heis3}
	\end{subfigure}
	\caption{\textbf{(Brockett integrator with agent interactions)}	Stochastic coverage of $\mathbb{R}^3$ by $N= 1,000$ agents at three 
 times $t$, following the semilinear PDE model \eqref{eq:clpPDEdsicon}. Blue agents are in the motion state; 
 red agents are in the motionless state.} 
	\label{fig:brockettint}
\end{figure*}

\begin{figure}
	\centering
		\includegraphics[scale=0.2]{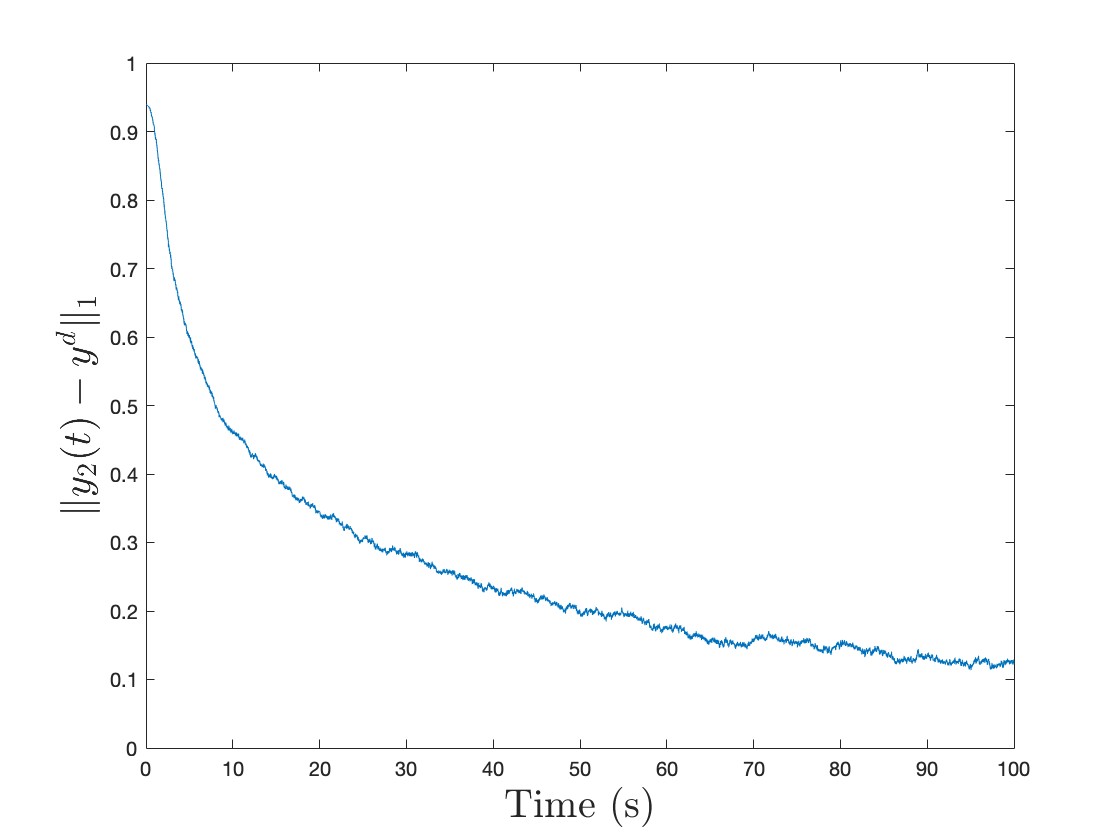}	
\caption{\kar{\textbf{(Brockett integrator with agent interactions)} Time evolution of the $L_1$ norm of the difference between the target distribution and the agent distribution \kar{in the motionless state} that evolves according to 
the semilinear PDE model \eqref{eq:clpPDEdsicon}.}} 
\label{fig:L1inin}
\end{figure}

The  positions  of $N_p=   1,000$ agents  are generated from a stochastic simulation of the SDE \eqref{eq:litSDE2} and plotted in  Fig. \ref{fig:brockettint} at three times $t$.  As  can be  seen  in this figure, 
at  time \kar{$t=100~s$}  the  swarm  is uniformly distributed over the sets $B_{\mathbf{x}_i}$ according to the target density. \kar{
\kar{Figure \ref{fig:L1inin} shows that} 
the $L_1$ norm of the difference between the current distribution \kar{of agents in the motionless state} and the target distribution decreases over time.} At \kar{$t=100~s$,} there are only $40$ agents in the state of motion. In contrast, when using the control approach without agent interactions (Section \ref{sec:subell}), all agents are constantly in motion. 
\kar{The $L_1$ norm at \kar{time $t = 100~s$}  is smaller for the case with \kar{agent} interactions \kar{(Fig. \ref{fig:L1inin})} than for the case without \kar{agent} interactions 
(Fig. \ref{fig:L1nonin}).}

While the interacting-agent control approach, unlike the control approach without interactions, enables agents to stop moving once the target density is reached (and therefore stop unnecessarily expending energy), 
the time until the interacting agents converge to the target density was found to be sensitive to the reaction constant $k$. Lower values of $k$, e.g. $k=10$, resulted in a slower rate of agent transitions to the motionless state. On the other hand,  if the value of $k$ was chosen too large, some of the agents prematurely transitioned to the motionless state in regions close to their initial positions.  
The performance of the interacting-agent control law was also affected by the parameter $\epsilon$. If $\epsilon$ was taken to be too small, for example $\epsilon = 0.1$, the agents did not converge to any distribution, but instead remained in a state of motion. This can be attributed to the fact that given a fixed value of $N_p$, the sum $c(\epsilon)\frac{1}{N_p}\sum_{i=1}^{N_p}K_{\epsilon}(\mathbf{x},\mathbf{x}_i(t))$ becomes a less accurate approximation of  the density $y(\mathbf{x},t)$ as the value of $\epsilon$ is decreased. On the other hand, if $\epsilon$ is taken to be too large, then the agent density converges to a regularized approximation  of the target density, rather than the target density itself. Due to space limitations, we do not include numerical results on the effects of these parameters here. 

\begin{example}
		\textbf{\kar{Underactuated} system on the sphere} 
\end{example}
		In this example, we consider a system on the $2$-dimensional sphere embedded in $\mathbb{R}^3$ given by $S^2 = \{\mathbf{x} \in \mathbb{R}^3;~ \mathbf{x}^T\mathbf{x} = 1\}$.
		\kar{We define	
  the following matrices $\mathbf{B}_i$, $i=1,2,3$:} 
   \begin{eqnarray}
\label{eq:genma}
\mathbf{B}_1 =\hspace{-1mm} 
\begin{bmatrix}
0 & -1 & 0 \\
1 &  0  & 0 \\
0 & 0  & 0 
\end{bmatrix}\hspace{-1mm},~
\mathbf{B}_2 =\hspace{-1mm}
\begin{bmatrix}
0 & 0& 1 \\
0 &  0  & 0 \\
-1 & 0  & 0 
\end{bmatrix}\hspace{-1mm},~
\mathbf{B}_3 =\hspace{-1mm}
\begin{bmatrix}
0 & 0& 0 \\
0 &  0  & -1 \\
0 & 1  & 0 
\end{bmatrix}\hspace{-1mm}. \nonumber 
\end{eqnarray}
Each of the matrices defines a vector field $\tilde{X}_i$ on $S^2$ given by 
	\begin{equation}
(\tilde{X}_i f)(\mathbf{x}) = \lim_{t \rightarrow 0}\frac{f(e^{t\mathbf{B}_i}\x) - f(\x)}{t}
\end{equation}
for all $\mathbf{x}\in S^2$ and all functions $f\in C^{\infty} (S^2)$. We assume that each agent can  control its motion along the vector fields ${\tilde{X}_1,\tilde{X}_2}$.  Note that in this case, the system is \kar{underactuated. This is because ${\rm span}\{\tilde{X}_1(\x), \tilde{X}_2(\x),[\tilde{X}_1(\x),\tilde{X}_2(\x)]\} =T_{\x} S^2$, where it can be verified that $[\tilde{X}_1,\tilde{X}_2](\x) = \tilde{X}_3(\x)$ for all $\x \in S^2$.} 
 
The kernel function is defined as
			\begin{equation}
		K_{\epsilon}(\mathbf{x},\mathbf{y}) =\begin{cases} 
		\exp{\frac{-1}{1-({\rm acos}(\mathbf{x}^T\mathbf{y})/\epsilon)^2}}~~{\rm if}~{\rm acos}(\mathbf{x}^T\mathbf{y})<\epsilon\\
			0 ~~{\rm otherwise}
			\end{cases}
			\end{equation}
	for all $\x ,\mathbf{y} \in S^2$.	The target density $y^d :S^2 \rightarrow \mathbb{R}_{\geq 0}$  (with respect to the Haar measure) is given by
		\begin{equation}
		y^d(\mathbf{x}) = 
		\begin{cases} 
		c ~~{\rm if} ~ x^2_i\geq 0.75~{\rm for} ~ i \in \{ 1,2,3\}\\
		0 ~~{\rm otherwise} 
		\end{cases}
		\end{equation}
		for all $\x ,\mathbf{y} \in S^2$, where $c$ is a normalization parameter chosen such that this function integrates to $1$.
		We set  $\epsilon=0.1$.

The  positions of $N_p=   1,000$ agents are generated from a stochastic simulation of the SDE \eqref{eq:litSDE2} and plotted in  Fig. \ref{fig:S2} at three times $t$. 
 The target density $y^d$ is depicted  on  the  surface  of  the  sphere  using  a color  density  plot. Blue  regions  are  assigned  a  low  target density of agents, while yellow regions are assigned a high target density. The agent positions are superimposed on   the  density  plot  to  enable  comparison between  the  actual  and  target densities.    Figure \ref{fig:S2} shows that  at  time $t= 100~s$,  the  distribution  of  the swarm over the sphere is close to the target density. As for the case of the Brockett integrator in Example \ref{eg:egbini2}, only a small fraction of the swarm ($62$ agents) is in the state of motion once the swarm has converged closely to the target density ($t = 100~s$). 

\begin{figure*}
	\centering
	\begin{subfigure}[t]{0.3\textwidth}
		\centering
		\includegraphics[width= \textwidth]{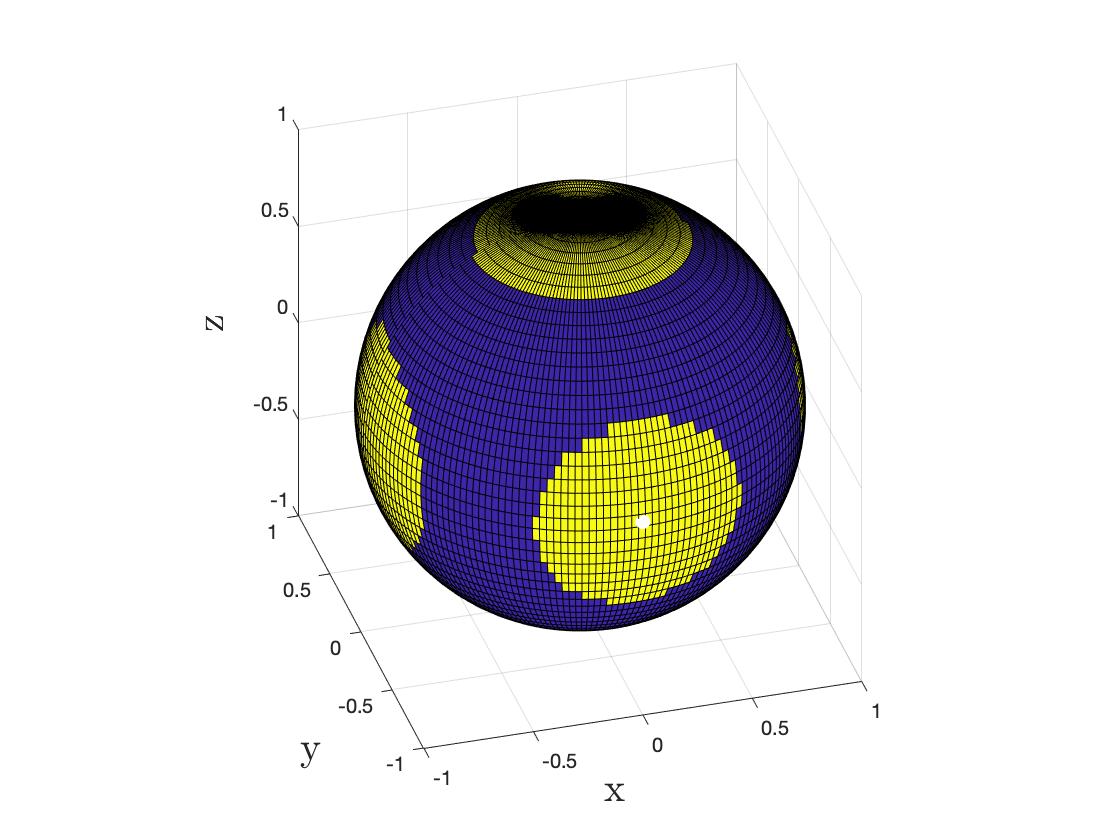}
		\caption{$t=0$ s}
		\label{subfig:Sphere1}
	\end{subfigure}%
	~
	\begin{subfigure}[t]{0.3\textwidth}
		\centering
		\includegraphics[width= \textwidth]{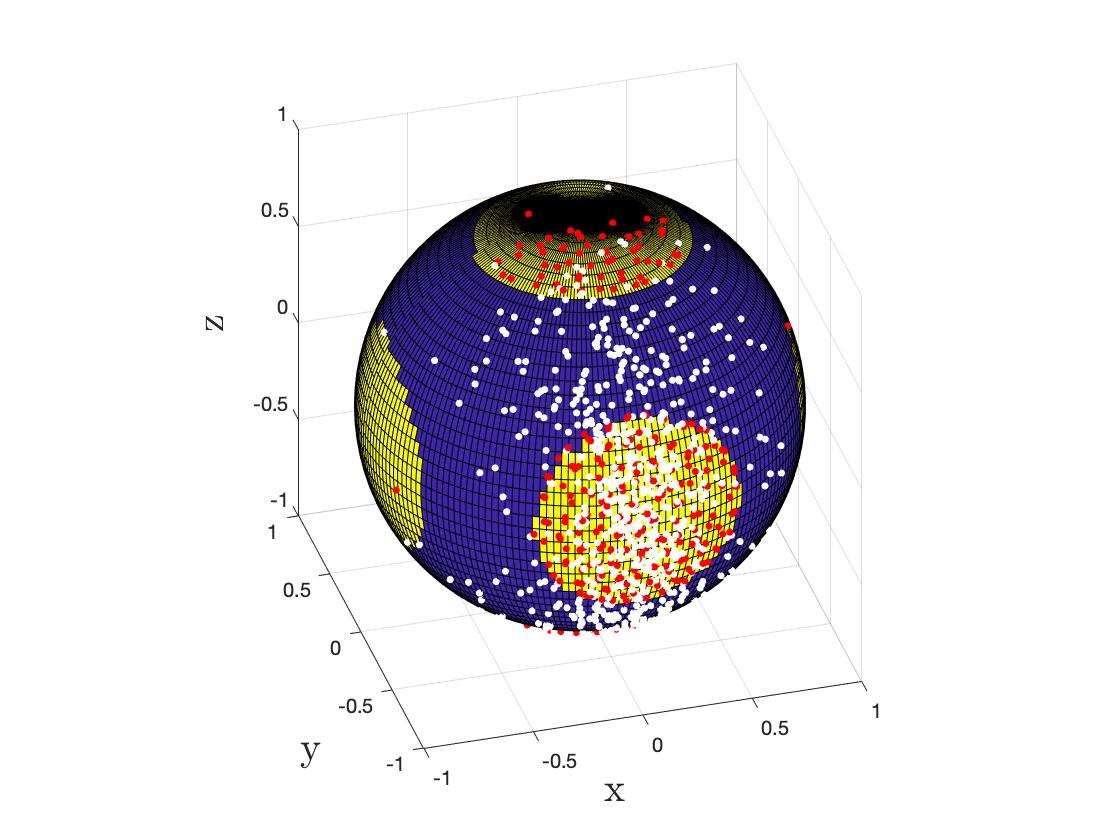}
		\caption{$t=10$ s}
			\label{subfig:Sphere3}
	\end{subfigure}
	~
	\begin{subfigure}[t]{0.3\textwidth}
		\centering
		\includegraphics[width= \textwidth]{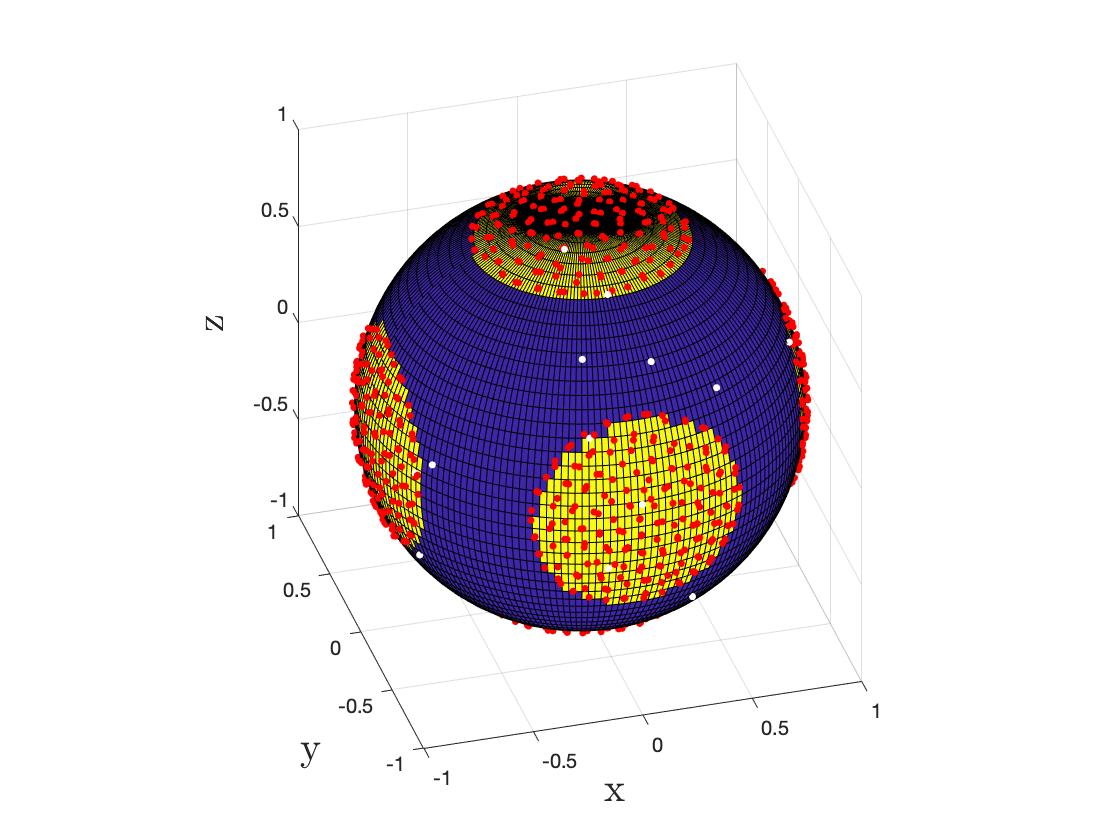}
		\caption{$t=100$ s}
				\label{subfig:Sphere3c}
	\end{subfigure}
	\caption{\textbf{(\kar{Underactuated} system on $S^2$ with agent interactions)}	Stochastic coverage of $S^2$ by $N= 1,000$ agents at three 
 times $t$, following the semilinear PDE model \eqref{eq:clpPDEdsicon}. White agents  are in the motion state; 
 red agents are in the motionless state.} 
\label{fig:S2}
\end{figure*}
	
	\section{CONCLUSION} \label{sec:conc}
	In this article, we have generalized our diffusion-based multi-agent coverage approach to the case where the agents have nonholonomic dynamics. We established exponential stability of the resulting Kolmogorov forward equation, whose generator is a hypoelliptic operator. In addition, we constructed a hybrid switching diffusion process of mean-field type such that the probability density of the random variable that represents the distribution of a swarm can be stabilized to a target density that is not necessarily positive everywhere on the domain. One possible direction for future work is to investigate 
 the tradeoffs between control laws with and without agent interaction. 
 \kar{Another is to incorporate pairwise
 interactions between agents that model collision avoidance maneuvers, which would require the inclusion of corresponding interaction terms in the PDE model.} 
One could also investigate the convergence of the $N$-agent system of hybrid switching diffusion processes to the solution of the semilinear PDE.

    \bibliographystyle{plain}
	\bibliography{cdcref_v2}
	\begin{IEEEbiography}[{\includegraphics[width=1in,height=1.25in,clip,keepaspectratio]{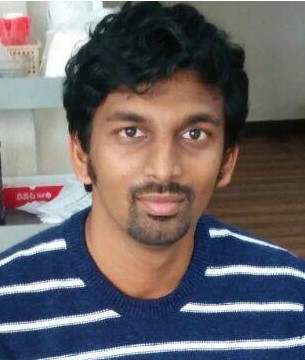}}]{Karthik Elamvazhuthi}
received the B.Tech. degree in mechatronics engineering from the Manipal Institute of Technology, Manipal, India in 2011. He received the M.S. and Ph.D. degrees in mechanical engineering from Arizona State University, Tempe, AZ, in 2014 and 2019, respectively.  He was a CAM Assistant Adjunct Professor at the Department of Mathematics, University of California, Los Angeles from 2019--2022. Currently, he is a postdoctoral scholar at the  Department of Mechanical Engineering, University of California, Riverside. His research interests include modeling and control of robotic swarms using methods from partial differential equations and stochastic processes.
\end{IEEEbiography}

\vspace{-13mm}
\begin{IEEEbiography}[{\includegraphics[width=1in,height=1.25in,clip,keepaspectratio]{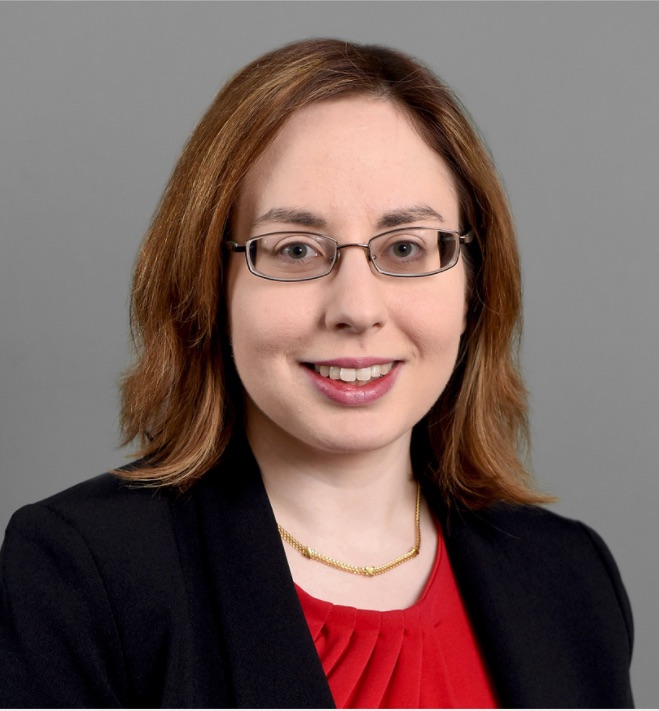}}]{Spring Berman}
(M'07) received the M.S.E. and
Ph.D. degrees in Mechanical Engineering and
Applied Mechanics from the University of Pennsylvania in 2008 and
2010, respectively.
From 2010 to 2012, she was a Postdoctoral Fellow in Computer Science at Harvard
University. She was then an Assistant Professor of Mechanical and Aerospace Engineering
at Arizona State University, where she has been an Associate Professor since 2018.
Her research focuses on the analysis of behaviors in biological and engineered collectives and the synthesis of control strategies for robotic swarms.
\end{IEEEbiography}
\end{document}